\documentclass[12pt]{article}
\usepackage{amsmath,amsthm,amssymb,mathtools}
\usepackage[utf8]{inputenc}
\usepackage{enumitem}
\usepackage{xcolor}
\usepackage[a4paper,margin=1in]{geometry}
\usepackage[authoryear,round]{natbib}
\usepackage[colorlinks=true,linkcolor=blue,citecolor=blue,urlcolor=blue]{hyperref}
\usepackage[font=small]{caption}
\usepackage{tikz}
\usetikzlibrary{positioning,fit,backgrounds}

\newtheorem{theorem}{Theorem}
\newtheorem{lemma}[theorem]{Lemma}
\newtheorem{corollary}[theorem]{Corollary}
\newtheorem{proposition}[theorem]{Proposition}
\DeclareMathOperator{\indeg}{indeg}
\DeclareMathOperator{\outdeg}{outdeg}
\DeclareMathOperator{\Cl}{Cl}
\theoremstyle{definition}
\newtheorem{definition}[theorem]{Definition}

\title{Minimum Network Level Forced by Hardwired Cluster Data}
\author{
Shilong Dai\thanks{Email: \texttt{daishilong@mails.ccnu.edu.cn}}
\qquad
Yangjing Long\thanks{Email: \texttt{yangjing@ccnu.edu.cn}}\\[0.5em]
\small School of Mathematics and Statistics, Central China Normal University,\\
\small No. 152, Luoyu Road, Wuhan, Hubei, PR China
}
\date{}

\begin{document}
\maketitle

\begin{abstract}
Reticulate evolutionary events, such as hybridization, recombination, and
horizontal transfer, can make a tree model inadequate. When evolutionary data
are summarized as hardwired clusters, one can ask how much local reticulation
complexity is forced by the data itself. We address this question for an
arbitrary cluster system $\mathcal C$ on a finite taxon set $X$ by computing
the minimum level of a rooted phylogenetic network whose hardwired cluster
system is exactly $\mathcal C$. Writing $H=\mathcal H[\mathcal C]$, we define
for each non-trivial block $B$ of $H$ a parameter $\mu(B)$ from generating sets
of incompatibility intersections in $B$. If $\ell(\mathcal C)$ denotes the
minimum level of any rooted network $N$ with $C_N=\mathcal C$, then
\[
\ell(\mathcal C)=\max\{\,\mu(B)\mid B\text{ is a non-trivial block of }H\,\}.
\]
Equivalently, $\mathcal C$ is realizable by a rooted level-$k$ network if and
only if $\mu(B)\le k$ for every non-trivial block $B$ of $H$. The lower-bound
proof relates incompatibility intersections to non-root hybrid vertices in
realizing blocks, while the upper-bound proof starts from the Hasse diagram and
iteratively splits selected hybrid vertices without changing the hardwired
cluster system. The result turns a network-design problem into a cluster-side
criterion and provides an interpretable complexity score for hardwired cluster
data, distinct from softwired cluster representation where clusters need only
occur in one displayed tree.
\end{abstract}

\par\medskip

\section{Introduction}

Reticulate evolutionary events such as hybridization, recombination, and
horizontal transfer can make a single tree an inadequate model for evolutionary
history.  Rooted phylogenetic networks extend rooted trees by allowing
reticulation vertices, and the level of a network measures this complexity
locally: it bounds the number of hybrid vertices that may occur inside one
non-tree-like block.  Level is therefore a natural mathematical proxy for the
local reticulate complexity that is needed to explain a collection of
evolutionary signals \citep{HusonRuppScornavacca2011,HusonScornavacca2011}.

One common way to summarize such signals is by clusters, that is,
descendant leaf sets supported by trees, networks, or gene-tree summaries.  In a
rooted phylogenetic tree, the cluster hierarchy determines the tree uniquely
\citep{SempleSteel2003}.  In a rooted network this uniqueness disappears
\citep{NakhlehWang2005}, and two different cluster semantics become relevant.
Under the softwired semantics, a cluster is represented if it occurs as a clade
in at least one tree displayed by the network.  Under the hardwired semantics,
a cluster must be the full descendant leaf set of an actual network vertex.
Thus softwired clusters are read after resolving reticulation choices, whereas
hardwired clusters are read directly from the network.  We follow the
softwired--hardwired terminology used in the cluster-representation literature
of \citet{NakhlehWang2005} and \citet{KelkScornavaccaVanIersel2011}, and the
hardwired cluster-system viewpoint developed by
\citet{HellmuthSchallerStadler2023}.

The softwired cluster model has led to a substantial algorithmic
literature.  Kelk, Scornavacca, and van Iersel studied the softwired
minimum-level problem for cluster sets, proved polynomial-time solvability for
each fixed level $k$, and analyzed the Cass algorithm, including cases where it
does not minimize level for general cluster sets
\citep{KelkScornavaccaVanIersel2011}.  Related work constructs cluster
networks for summarizing gene-tree information \citep{HusonRupp2008}, uses
fewer reticulations to represent conflicting clusters
\citep{vanIerselKelkRuppHuson2010}, studies bounded-level encodings
\citep{GambetteHuber2012}, and controls reticulation complexity through
decompositions or explicit constructions
\citep{GusfieldBansalBafnaSong2007,WangGuoLiuEtAl2013,ToHabib2009}.  These
results are complementary to the question treated here, because softwired
representation usually asks the network to contain the input clusters, possibly
along with additional clusters.  The hardwired problem considered in this paper
prescribes the whole hardwired cluster system exactly.

Our motivating question is therefore the following.  Given hardwired
cluster information extracted from evolutionary data, what is the minimum local
reticulation complexity forced by the data itself?  Equivalently, for a cluster
system $\mathcal C$ on a finite taxon set $X$, when does there exist a rooted
level-$k$ network $N$ on $X$ whose hardwired cluster system is exactly
$C_N=\mathcal C$?  We seek a criterion that can be computed from the cluster
system before constructing, or guessing, any realizing network.

The Hasse diagram $H=\mathcal H[\mathcal C]$ is the natural object in
which to ask this question.  Previous work has shown that Hasse diagrams,
blocks, and incompatibility structure are central for hardwired cluster systems
of rooted networks \citep{HellmuthSchallerStadler2023}; results on regular and
related network classes also give reconstruction statements in restricted
settings \citep{BaroniSempleSteel2004,Willson2010,Zhang2019}.  What is missing
is a necessary-and-sufficient level-$k$ criterion, expressed directly in
$\mathcal H[\mathcal C]$, for an arbitrary prescribed hardwired cluster system.
The broader cluster-based viewpoint also appears in work on LCA-cluster
systems in directed acyclic graphs
\citep{ShanavasChangatHellmuthStadler2023}, common refinements of rooted
phylogenetic trees \citep{SchallerHellmuthStadler2021}, cluster containment
and network classes \citep{GunawanLuZhang2018,Zhang2019}, and structural
properties of phylogenetic networks
\citep{Hayamizu2018,SuzukiHayamizu2025}.

We provide such a criterion.  For every non-trivial block $B$ of
$H=\mathcal H[\mathcal C]$, we define a parameter $\mu(B)$ from the
intersections of incompatible clusters in $B$ and the smallest cluster pieces
needed to generate those intersections.  Our main theorem proves
\[
\ell(\mathcal C)=
\max\{\,\mu(B)\mid B\text{ is a non-trivial block of }\mathcal H[\mathcal C]\,\},
\]
where $\ell(\mathcal C)$ is the minimum level of a rooted network realizing
$\mathcal C$ exactly.  Thus $\mathcal C$ is realizable by a rooted level-$k$
network if and only if $\mu(B)\le k$ for every non-trivial Hasse block $B$.
The lower-bound direction shows that incompatibility-intersection pieces force
non-root hybrid vertices in any realizing block.  The upper-bound direction
starts from the Hasse diagram and splits selected hybrid vertices without
changing the hardwired cluster system, showing that no larger blockwise level
is needed.  Figure~\ref{fig:roadmap} summarizes this route from the input
cluster system to the minimum-level formula.

\begin{figure}[t]
\centering
\begin{tikzpicture}[
  >=stealth,
  font=\scriptsize,
  box/.style={draw=blue!55!black, rounded corners=2pt, align=center,
    inner xsep=8pt, inner ysep=6pt, minimum height=0.95cm, fill=white},
  main/.style={box, text width=4.8cm},
  side/.style={box, text width=5.2cm, minimum height=1.55cm},
  formula/.style={box, text width=5.5cm, fill=blue!4, minimum height=1.05cm},
  arrow/.style={->, thick, draw=blue!60!black}
]
\node[main] (input) at (0,0) {Input\\ cluster system $\mathcal C$};
\node[main] (hasse) at (0,-1.55) {Hasse diagram\\ $H=\mathcal H[\mathcal C]$};
\node[side] (lower) at (-4.15,-3.85) {\textbf{Lower bound}\\
tools: Props.~\ref{prop:incompatibility-facts-realizing-blocks},
\ref{prop:hasse-block-realized-in-one-network-block},
\ref{prop:network-generators-localize}\\
any realization gives a generating set};
\node[formula] (combine) at (0,-6.25) {Combine the two bounds\\
  $\displaystyle \ell(\mathcal C)=\max_B \mu(B)$};
\node[side] (upper) at (4.15,-3.85) {\textbf{Upper bound}\\
tools: Props.~\ref{prop:admissible-bad-hybrid-splits},
\ref{prop:splitting-invariants},
\ref{prop:block-refinement},
\ref{prop:termination-final-block-control}\\
split redundant hybrids starting from $H$};

\draw[arrow] (input) -- (hasse);
\draw[arrow] (hasse) -- (lower);
\draw[arrow] (hasse) -- (upper);
\draw[arrow] (lower) -- (combine);
\draw[arrow] (upper) -- (combine);
\end{tikzpicture}
\caption{Proof architecture for the minimum-level formula.  The lower-bound
package shows that each conflict block forces at least $\mu(B)$ non-root
hybrid labels in any realization, while the upper-bound package constructs a
realization by splitting redundant hybrid vertices in the Hasse diagram.}
\label{fig:roadmap}
\end{figure}
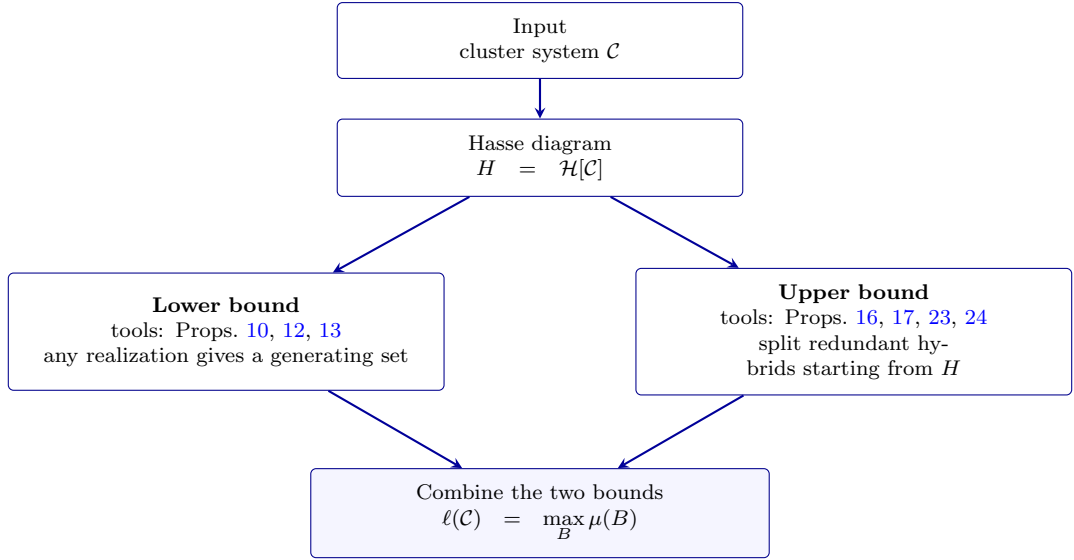

The paper is organized as follows.  Section~2 introduces cluster systems,
Hasse diagrams, blocks, network level, and the block parameter $\mu(B)$, gives a
set-cover formulation, and develops the examples.  Section~3 states the main
results, explains why $\mu(B)$ is the correct threshold, and records the
computational consequences.  Section~\ref{sec:proofs-main} keeps the proofs of
the main results in the body of the paper, so that the logical route from
hardwired cluster data to the minimum-level formula remains visible.  The most
technical ingredients are then placed in the appendix: Appendix~\ref{app:lower-bound}
contains the lower-bound tools, while Appendices~\ref{app:splitting}--\ref{app:termination}
contain the upper-bound construction: the splitting operation, the cut-vertex
and block-localization controls, and the termination and final block-control
argument.

\section{Hasse Blocks and the Incompatibility-Intersection Parameter}

We fix notation for cluster systems, Hasse diagrams, rooted networks, and the
block parameters used in the characterization.  This section is intentionally
self-contained: notation that is used throughout the proof is introduced here,
rather than inside later lemma statements.  Let $X$ be a finite nonempty set,
let $\mathcal C$ be a cluster system on $X$, and write
$H:=\mathcal H[\mathcal C]$.

\subsection{Cluster systems, Hasse diagrams, and compatibility}

\paragraph{Rooted networks and clusters.}
We use standard rooted-network terminology
\citep{HusonRuppScornavacca2011,HellmuthSchallerStadler2023}.  A rooted
network is a finite directed acyclic graph $N=(V,E)$ with a unique root
$\rho_N$ of indegree $0$.  A vertex $v$ is a leaf if $\outdeg_N(v)=0$, a
hybrid, or reticulation, vertex if $\indeg_N(v)>1$, and a tree vertex if
$\indeg_N(v)\le 1$.  The leaf set is denoted by $X(N)$; if $X(N)=X$, then $N$
is a rooted network on $X$.  Throughout, network means rooted network.  For
$v\in V(N)$, let $\sigma(v)$ be the set of leaves
reachable from $v$; the hardwired cluster system of $N$ is
$C_N:=\{\sigma(v)\mid v\in V(N)\}$. This is the vertex-based hardwired
convention used for cluster systems of rooted networks in
\citet{HellmuthSchallerStadler2023}.

For comparison, the softwired cluster system is obtained by ranging over
trees displayed by $N$ and recording the clusters that occur as clades in those
trees.  Equivalently, one may choose one incoming edge at each hybrid vertex,
delete the other incoming edges, suppress degree-two vertices in the usual way,
and then read clusters from the resulting displayed tree.  The two conventions
coincide for rooted trees but diverge for networks: a hardwired cluster is a
descendant set of an actual network vertex, whereas a softwired cluster may
appear only after a particular resolution of reticulation choices.  All
realization results in this paper use the hardwired convention and require
equality $C_N=\mathcal C$.

A cluster on $X$ is a nonempty subset of $X$.  A cluster system on $X$ is a
family $\mathcal C\subseteq 2^X\setminus\{\varnothing\}$ such that
$X\in\mathcal C$ and $\{x\}\in\mathcal C$ for every $x\in X$.

\paragraph{Hasse diagrams and order notation.}
The Hasse diagram $\mathcal H[\mathcal C]$ is the directed acyclic graph with
vertex set $\mathcal C$ and edge $C\to C'$ precisely when $C'\subsetneq C$ and
there is no $C''\in\mathcal C$ with $C'\subsetneq C''\subsetneq C$.  For a
rooted network $N$ and vertices $u,v$, write $v\preceq_N u$ if there is a
directed path from $u$ to $v$, and write $v\prec_N u$ if $v\preceq_N u$ and
$v\neq u$.  Vertices $u$ and $v$ are comparable if $u\preceq_N v$ or
$v\preceq_N u$, and incomparable otherwise.  If $u\to v$ is an edge of $N$,
then $u$ is a parent of $v$ and $v$ a child of $u$.  We write
$u\leadsto v$ for a directed path from $u$ to $v$, possibly of length zero.

We freely identify a cluster $A\in\mathcal C$ with the corresponding vertex of
$H$ whenever convenient. Thus expressions such as $A\in V(B_H)$, $A\subseteq A'$,
and $A\in\Cl(B_N)$ refer to the same cluster viewed respectively as a
vertex of the Hasse diagram, as a subset of $X$, and as an element of a cluster
family. The graph $H$ is connected, has $X$ as its unique root, its
inclusion-minimal vertices are exactly the singletons, and every non-singleton
cluster $C\in\mathcal C$ satisfies $\outdeg_H(C)\ge 2$; hence $H$ is itself a
rooted network. Later, for current networks arising in the splitting
construction, $\lambda(v)$ denotes the original Hasse-label of a current vertex
$v$, so that $\lambda(v)\in V(B)$ means that this original label belongs to the
original non-trivial block $B$ of $H$.

\paragraph{Compatibility, incompatibility, and blocks.}
We use the standard compatibility terminology for clusters
\citep{SempleSteel2003,HusonRuppScornavacca2011}.  Two clusters
$A,A'\in\mathcal C$ are \emph{compatible} if $A\cap A'=\varnothing$,
$A\subseteq A'$, or $A'\subseteq A$.  They are \emph{incompatible}
if
$A\cap A'\neq\varnothing$, $A\not\subseteq A'$, and
$A'\not\subseteq A$.  The graph whose vertices are the clusters in a
family and whose edges join incompatible pairs is the usual
\emph{incompatibility graph}; to avoid overloading notation, the symbol
$\mathcal J(B)$ below denotes a family of intersections, not this graph.  A
block, or biconnected component, of a rooted network is a maximal
$2$-connected subgraph of its underlying undirected graph.  A block is
non-trivial if it contains an undirected cycle, equivalently, if it is neither
a single vertex nor a single edge.  For a non-trivial block $B$ of a rooted
network, a vertex $m\in V(B)$ is terminal if it is $\preceq_N$-minimal in
$V(B)$.

For a connected undirected graph $G$, its block-cut tree $T(G)$ is the
bipartite graph whose nodes are the blocks and cut vertices of $G$, with a
block-node $B$ adjacent to a cut-vertex-node $c$ exactly when $c\in V(B)$.

\subsection{Blocks and network level}

\paragraph{Block notation convention.}
When Hasse-diagram blocks and network blocks appear in the same argument, we
write $B_H,D_H$ for blocks of $H=\mathcal H[\mathcal C]$ and $B_N,D_N$ for
blocks of a realizing or current network.  In sections dealing only with the
splitting construction, an unadorned $B$ is an original block of $H$ unless a
block is explicitly called current or network-side.  For an original cut vertex,
$B_H^\uparrow$ denotes the unique original block in which the cut vertex is a
non-root vertex, and $B_H^\downarrow$ denotes a child block below that cut
vertex in the rooted block-cut tree.

If $M$ is a rooted network with underlying undirected graph $G$, we root $T(G)$
at the node corresponding to the root of $M$: if $\rho_M$ is a cut vertex of
$G$, we take the cut-vertex-node $\rho_M$ as the root; otherwise we take the
unique block-node containing $\rho_M$. The resulting rooted tree is the
\emph{rooted block-cut tree} of $M$. The same convention is used when
$M=H=\mathcal H[\mathcal C]$. In this rooted tree we speak of ancestor and
descendant blocks, and a topological order of the blocks means any order
compatible with the descendant relation.

The next two lemmas record basic structural properties of non-trivial blocks in rooted networks that will be used repeatedly in the sequel.

\begin{lemma}
Let $N$ be a rooted network, and let $B$ be a non-trivial block of $N$.
Then $V(B)$ has a unique $\preceq_N$-maximal vertex.
\end{lemma}

\begin{proof}
Let $r_B$ be the root of the block-node $B$ in the rooted block-cut tree of $N$:
if $B$ is the root block, set $r_B:=\rho_N$; otherwise let $r_B$ be the cut
vertex corresponding to the parent of $B$. Then $r_B\in V(B)$.

For any $x\in V(B)$, every directed path from $\rho_N$ to $x$ enters the
block-node $B$ through $r_B$, so $x\preceq_N r_B$. Thus $r_B$ is
$\preceq_N$-maximal in $V(B)$. If $y\in V(B)$ is another maximal vertex, then
$y\preceq_N r_B$, hence $y=r_B$. Therefore the maximal vertex is unique.
\end{proof}

For a non-trivial block $B$ of a rooted network, write $\max B$ for this unique
$\preceq_N$-maximal vertex and call it the \emph{root} of $B$. A vertex of $B$
is a \emph{non-root vertex} of $B$ if it is distinct from $\max B$; if such a
vertex is hybrid, it is a \emph{non-root hybrid vertex} of $B$. Since $H$ is
itself a rooted network, we use the same notation for non-trivial blocks of $H$.

\paragraph{Interior of a block.}
Let $N$ be a rooted network, and let $B$ be a non-trivial block of $N$.  If
$m_1,\dots,m_h$ are the terminal vertices of $B$, define
$B^0:=V(B)\setminus\bigl(\{\max B\}\cup\{m_1,\dots,m_h\}\bigr)$.

\begin{lemma}
\label{lem:terminal-two-parents}
Let $N$ be a rooted network, let $B$ be a non-trivial block of $N$, and let $m\in V(B)$
be a terminal vertex of $B$.
Then $m$ has at least two distinct parents in $B$.
\end{lemma}

\begin{proof}
Since the underlying undirected graph of $B$ is biconnected, every vertex of $B$
has degree at least $2$. Thus $m$ has distinct neighbors $a,b\in V(B)$. As $m$
is terminal in $B$, neither $m\to a$ nor $m\to b$ can occur, since that would
give $a\prec_N m$ or $b\prec_N m$. Hence $a\to m$ and $b\to m$, so $m$ has at
least two distinct parents in $B$.
\end{proof}

\paragraph{Network-side level.}
Following \citet{HellmuthSchallerStadler2023}, for a non-trivial block $B_N$ of
a rooted network $N$, let $H^\times(B_N)$ be the set of non-root hybrid vertices
of $B_N$ and put $h(B_N):=|H^\times(B_N)|$.  A rooted network $N$ is level-$k$
if every non-trivial block $B_N$ of $N$ satisfies $h(B_N)\le k$.

The minimum level of the cluster system $\mathcal C$ is
\[
\ell(\mathcal C):=
\min\{\,k\ge 0 \mid C_N=\mathcal C
\text{ for some rooted level-}k\text{ network }N\text{ on }X\,\}.
\]
This minimum is finite: after identifying each singleton vertex $\{x\}$ of
$H=\mathcal H[\mathcal C]$ with the taxon $x$, the Hasse diagram itself is a
rooted network on $X$ with hardwired cluster system $\mathcal C$.

\subsection{The block invariant \texorpdfstring{$\mu(B)$}{mu(B)}}

\begin{definition}
Let $B$ be either a non-trivial block of the Hasse diagram
$H=\mathcal H[\mathcal C]$ or a non-trivial block of a rooted network $N$. Define
\[
\Cl(B):=
\begin{cases}
V(B), & \text{if } B \text{ is a block of } H,\\[1mm]
\{\sigma(v)\mid v\in V(B)\}, & \text{if } B \text{ is a block of } N,
\end{cases}
\]
and
\[
\mathcal J(B):=
\{\,A\cap A' \mid A,A'\in \Cl(B),\ A \text{ and } A' \text{ are incompatible}\,\}.
\]
A set $K\subseteq \Cl(B)$ is an \emph{incompatibility-intersection generating
set} of $B$ if
\[
I=\bigcup\{\,E\in K\mid E\subseteq I\,\}
\qquad\text{for every } I\in\mathcal J(B).
\]
For brevity, we also call such a set a \emph{generating set} of $B$. Define
$\mu(B):=0$ if $\mathcal J(B)=\varnothing$, and otherwise
\[
\mu(B):=
\min\{\,|K| \mid K\subseteq \Cl(B)\text{ is a generating set of }B\,\}.
\]
If no generating set exists, this minimum is interpreted as $\infty$.
When needed, we write $\Cl(B_H),\mathcal J(B_H),\mu(B_H)$ for a
non-trivial block $B_H$ of $H$, and $\Cl(B_N),\mathcal J(B_N),\mu(B_N)$
for a non-trivial block $B_N$ of a rooted network.  Thus $\Cl(B)$ is always a
family of clusters, that is, a family of nonempty subsets of $X$; in the Hasse
case the vertices of $B$ are already such subsets.
\end{definition}

\begin{proposition}[Set-cover formulation]\label{prop:set-cover-formulation}
Let $B$ be a non-trivial Hasse block or a non-trivial network block, and define
$U_B:=\{\,(I,x)\mid I\in\mathcal J(B),\ x\in I\,\}$.
For $E\in\Cl(B)$, say that $E$ covers $(I,x)\in U_B$ if
$x\in E\subseteq I$.
Then $K\subseteq\Cl(B)$ is a generating set of $B$ if and only if $K$ covers
all elements of $U_B$. Consequently, $\mu(B)$ is the optimum value of this
finite set-cover instance, with optimum $0$ when $U_B=\varnothing$ and
optimum $\infty$ when no cover exists.
\end{proposition}

\begin{proof}
The defining equality for a generating set is equivalent to the following
pointwise condition: for every $I\in\mathcal J(B)$ and every $x\in I$, there is
some $E\in K$ such that $x\in E\subseteq I$. This is exactly the assertion that
the sets in $K$ cover every element $(I,x)$ of $U_B$.

Indeed, if the equality
\(I=\bigcup\{\,E\in K\mid E\subseteq I\,\}\)
holds, then each $x\in I$ lies in some member $E\in K$ with $E\subseteq I$.
Conversely, if this pointwise condition holds, every element of $I$ lies in the
union on the right, while every set contributing to that union is contained in
$I$. Hence the two sides are equal. Since $\Cl(B)$ and $\mathcal J(B)$ are
finite, this is a finite set-cover instance. The stated conventions agree with
the definition of $\mu(B)$.
\end{proof}

\paragraph{Intuition.}
If two clusters in $B$ are incompatible, their intersection is a shared leaf-set region
that has to be accounted for inside the block.  A generating set $K$ records
cluster pieces from $B$ whose unions reconstruct all such regions.  Thus
$\mu(B)$ is the minimum number of pieces needed to reconstruct all
incompatibility intersections in $B$.  In a realizing network, the lower-bound argument
charges these pieces to descendant clusters of non-root hybrid vertices, so
$\mu(B)$ measures the reticulation structure forced by the cluster system.

\subsection{\texorpdfstring{Biological interpretation of $\mu(B)$}{Biological interpretation of mu(B)}}

The parameter $\mu(B)$ has a direct interpretation in terms of local
conflict in hardwired cluster data.  A non-trivial Hasse block $B$ collects
clusters that cannot be separated into independent tree-like pieces; it is the
cluster-side analogue of one local tangle in a phylogenetic network.  When two
clusters in $B$ are incompatible, their intersection is a descendant leaf
region that is demanded simultaneously by different supported groupings.  Such
regions are precisely where a purely tree-like explanation breaks down.

The generating set in the definition of $\mu(B)$ asks how many cluster
pieces are needed to account for all these shared descendant regions.  Thus
$\mu(B)$ is not the number of hybrid vertices that happen to appear in the
Hasse diagram; those vertices may be redundant artifacts of the representation.
Instead, $\mu(B)$ counts the minimum number of independent descendant-region
pieces that any exact hardwired realization must explain by non-root hybrid
vertices.  Consequently
$\max_B\mu(B)$ is the minimum local reticulation complexity forced by the
hardwired cluster data, block by block.

\subsection{Examples}

We give two small examples illustrating the block parameter $\mu(B)$. In the
figures, thick solid arcs belong to the indicated non-trivial block, dashed gray
arcs are pendant edges outside that block, and filled vertices form a minimum
incompatibility-intersection generating set $K$.

\tikzset{
    v/.style={circle,draw,fill=white,inner sep=1.2pt,minimum size=4.5pt},
    gen/.style={circle,draw,fill=black,inner sep=1.2pt,minimum size=4.5pt},
    hyb/.style={circle,draw,fill=gray!25,inner sep=1.2pt,minimum size=4.5pt},
    bedge/.style={->,line width=0.75pt},
    pedge/.style={->,gray,densely dashed,line width=0.45pt},
    lab/.style={font=\scriptsize}
}

\smallskip
\noindent\textbf{Example 1: a block with $\mu(B)=1$.}

Let $X=\{1,2,3\}$ and
$\mathcal C_1=\{X,A=\{1,2\},C=\{2,3\},\{1\},\{2\},\{3\}\}$.
The Hasse diagram has one non-trivial block
$B_1=\{X,A,C,\{2\}\}$.

\begin{figure}[ht]
\centering
\begin{tikzpicture}[scale=1.05,>=stealth]
\node[v,label={[lab]above:$X$}] (X) at (0,0) {};
\node[v,label={[lab]left:$A$}] (A) at (-1.4,-1.1) {};
\node[v,label={[lab]right:$C$}] (C) at (1.4,-1.1) {};
\node[v,label={[lab]below:$\{1\}$}] (one) at (-2.1,-2.2) {};
\node[gen,label={[lab]below:$\{2\}$}] (two) at (0,-2.2) {};
\node[v,label={[lab]below:$\{3\}$}] (three) at (2.1,-2.2) {};

\draw[bedge] (X) -- (A);
\draw[bedge] (X) -- (C);
\draw[pedge] (A) -- (one);
\draw[bedge] (A) -- (two);
\draw[bedge] (C) -- (two);
\draw[pedge] (C) -- (three);

\node[lab] at (0.95,-0.15) {$B_1$};
\end{tikzpicture}
\caption{The unique non-trivial block of $\mathcal C_1$. The filled vertex is $K=\{\{2\}\}$.}
\label{fig:example-mu-one}
\end{figure}
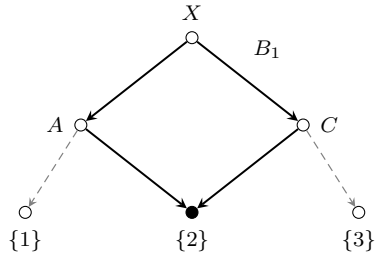

Inside $B_1$, the only incompatible pair is $A$ and $C$, with
$A\cap C=\{2\}$. Thus $\mathcal J(B_1)=\{\{2\}\}$. The set
$K=\{\{2\}\}$ is forced, since the only nonempty cluster in
$\Cl(B_1)$ contained in $\{2\}$ is $\{2\}$ itself. Hence
$\mu(B_1)=1$. The Hasse diagram itself realizes $\mathcal C_1$, and the
only non-root hybrid vertex in $B_1$ is $\{2\}$. This is the basic level-$1$
case.

\smallskip
\noindent\textbf{Example 2: illustrative hardwired cluster data with two forced reticulations.}

Suppose four taxa are summarized by a collection of well-supported
descendant leaf sets.  The following cluster system is meant only as a toy
example of such hardwired cluster data, not as an empirical data set. Let $X=\{1,2,3,4\}$ and
\[
\mathcal C_2=
\begin{aligned}[t]
\{&X,
P=\{1,2,3\},
Q=\{1,2,4\},
R=\{1,2\},
S=\{1,3\},
T=\{2,4\},
\\
&\{1\},\{2\},\{3\},\{4\}\}.
\end{aligned}
\]
The unique non-trivial block is
$
B_2=\{X,P,Q,R,S,T,\{1\},\{2\}\}.
$
The edges from $S$ to $\{3\}$ and from $T$ to $\{4\}$ are pendant edges outside
this block.

\begin{figure}[ht]
\centering
\begin{tikzpicture}[scale=1.0,>=stealth]
\node[v,label={[lab]above:$X$}] (X) at (0,0) {};

\node[v,label={[lab]left:$P$}] (P) at (-1.7,-1.05) {};
\node[v,label={[lab]right:$Q$}] (Q) at (1.7,-1.05) {};

\node[v,label={[lab]left:$S$}] (S) at (-3.0,-2.15) {};
\node[hyb,label={[lab]above:$R$}] (R) at (0,-2.15) {};
\node[v,label={[lab]right:$T$}] (T) at (3.0,-2.15) {};

\node[v,label={[lab]below:$\{3\}$}] (three) at (-3.0,-3.25) {};
\node[gen,label={[lab]below:$\{1\}$}] (one) at (-1.2,-3.25) {};
\node[gen,label={[lab]below:$\{2\}$}] (two) at (1.2,-3.25) {};
\node[v,label={[lab]below:$\{4\}$}] (four) at (3.0,-3.25) {};

\draw[bedge] (X) -- (P);
\draw[bedge] (X) -- (Q);

\draw[bedge] (P) -- (S);
\draw[bedge] (P) -- (R);
\draw[bedge] (Q) -- (R);
\draw[bedge] (Q) -- (T);

\draw[bedge] (S) -- (one);
\draw[bedge] (R) -- (one);
\draw[bedge] (R) -- (two);
\draw[bedge] (T) -- (two);

\draw[pedge] (S) -- (three);
\draw[pedge] (T) -- (four);

\node[lab] at (0.85,-0.3) {$B_2$};
\end{tikzpicture}
\caption{A block with three Hasse-hybrid vertices but only two forced generators. The gray vertex $R=\{1,2\}$ is hybrid but not needed in a minimum generating set; the black vertices form $K=\{\{1\},\{2\}\}$.}
\label{fig:example-hasse-three-hybrids}
\end{figure}
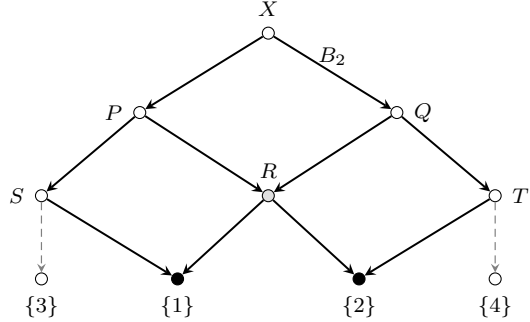

Viewing the Hasse diagram itself as a rooted network, the non-root hybrid
vertices in $B_2$ are $R=\{1,2\}$, $\{1\}$, and $\{2\}$.
Indeed, $R$ has parents $P$ and $Q$, the vertex $\{1\}$ has parents $R$ and
$S$, and the vertex $\{2\}$ has parents $R$ and $T$. Thus the Hasse diagram has
three non-root hybrid vertices inside this block.

However, the incompatibility intersections inside $B_2$ are generated by only two
clusters. The relevant intersections are \(P\cap Q=R=\{1,2\}\),
\(P\cap T=\{2\}\), and \(Q\cap S=\{1\}\),
together with the repeated intersections
$
R\cap S=\{1\},\qquad R\cap T=\{2\}.
$
Hence
$
\mathcal J(B_2)=\{\{1,2\},\{1\},\{2\}\}.
$
The set $K=\{\{1\},\{2\}\}$
generates every member of $\mathcal J(B_2)$, because
\(\{1,2\}=\{1\}\cup\{2\}\), \(\{1\}=\{1\}\), and \(\{2\}=\{2\}\).
Neither singleton can be omitted, so $K$ is minimum and
$
\mu(B_2)=2.
$
Thus this block satisfies
$
\mu(B_2)=2<3,
$
even though the Hasse diagram has three non-root hybrid vertices in the block.
This example illustrates why $\mu(B)$ measures the number of independent
incompatibility-intersection pieces, not simply the number of hybrid vertices already
present in the Hasse diagram.

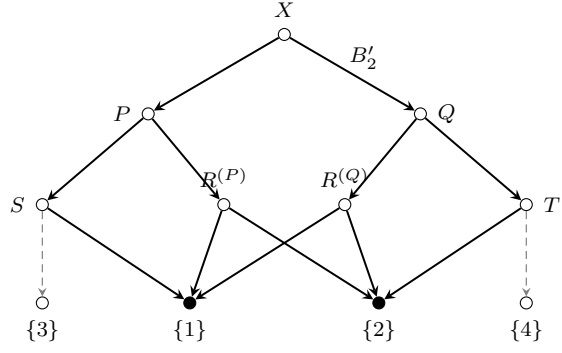
\begin{figure}[ht]
\centering
\begin{tikzpicture}[scale=1.0,>=stealth]
\node[v,label={[lab]above:$X$}] (X) at (0,0) {};

\node[v,label={[lab]left:$P$}] (P) at (-1.8,-1.05) {};
\node[v,label={[lab]right:$Q$}] (Q) at (1.8,-1.05) {};

\node[v,label={[lab]left:$S$}] (S) at (-3.2,-2.25) {};
\node[v,label={[lab]above:$R^{(P)}$}] (RP) at (-0.8,-2.25) {};
\node[v,label={[lab]above:$R^{(Q)}$}] (RQ) at (0.8,-2.25) {};
\node[v,label={[lab]right:$T$}] (T) at (3.2,-2.25) {};

\node[v,label={[lab]below:$\{3\}$}] (three) at (-3.2,-3.55) {};
\node[gen,label={[lab]below:$\{1\}$}] (one) at (-1.25,-3.55) {};
\node[gen,label={[lab]below:$\{2\}$}] (two) at (1.25,-3.55) {};
\node[v,label={[lab]below:$\{4\}$}] (four) at (3.2,-3.55) {};

\draw[bedge] (X) -- (P);
\draw[bedge] (X) -- (Q);

\draw[bedge] (P) -- (S);
\draw[bedge] (P) -- (RP);

\draw[bedge] (Q) -- (RQ);
\draw[bedge] (Q) -- (T);

\draw[bedge] (S) -- (one);
\draw[pedge] (S) -- (three);

\draw[bedge] (RP) -- (one);
\draw[bedge] (RP) -- (two);

\draw[bedge] (RQ) -- (one);
\draw[bedge] (RQ) -- (two);

\draw[bedge] (T) -- (two);
\draw[pedge] (T) -- (four);

\node[lab] at (1.05,-0.3) {$B_2'$};
\end{tikzpicture}
\caption{A level-$2$ network realizing the same cluster system $\mathcal C_2$. The original Hasse-hybrid vertex $R=\{1,2\}$ has been split into two vertices $R^{(P)}$ and $R^{(Q)}$, both with descendant cluster $\{1,2\}$. The only non-root hybrid vertices in the non-trivial block are the filled vertices $\{1\}$ and $\{2\}$.}
\label{fig:example-level-two-realization}
\end{figure}

The network in Figure~\ref{fig:example-level-two-realization} realizes the
same cluster system as the Hasse diagram in
Figure~\ref{fig:example-hasse-three-hybrids}. Indeed, the two vertices
$R^{(P)}$ and $R^{(Q)}$ both have descendant cluster $\{1,2\}$, so the split
does not introduce a new cluster. The descendant clusters of $P,Q,S,T$ remain
$P,Q,S,T$, respectively, and the leaves still give the four singletons. Hence
the hardwired cluster system is exactly $\mathcal C_2$.

In this realization, $R^{(P)}$ and $R^{(Q)}$ are not hybrid vertices, since
each has indegree one. The only non-root hybrid vertices in the non-trivial
block are $\{1\}$ and $\{2\}$. Therefore this realizing network has
$h(B_2')=2$, and so it is level-$2$.
Therefore any network that exactly realizes these hardwired clusters must
contain at least two non-root hybrid vertices in the relevant local block, and
Figure~\ref{fig:example-level-two-realization} shows that two are sufficient.
In this sense the data force level $2$, even though the original Hasse diagram
contains three non-root hybrid vertices in the same block.

\begin{proposition}[A Boolean-lattice gap]\label{prop:boolean-lattice-gap}
Let $n\ge 3$, let $X_n=\{1,\ldots,n\}$, and let
$\mathcal C_n=2^{X_n}\setminus\{\varnothing\}$.
The Hasse diagram $H_n=\mathcal H[\mathcal C_n]$ has a unique non-trivial
block $B_n$, namely the whole Hasse diagram. Moreover, $\mu(B_n)=n$, whereas
the Hasse diagram itself has $2^n-n-2$ non-root hybrid vertices in this block.
\end{proposition}

\begin{proof}
The underlying undirected graph of $H_n$ is the $n$-cube with the vertex
$\varnothing$ deleted. We first show that this graph has no cut vertex. Delete
a vertex $S\ne X_n$. Every remaining vertex $A$ is connected to $X_n$ while
avoiding $S$: if $A\not\subseteq S$, add the missing elements of $X_n$ one at a
time; if $A\subseteq S$, first add some element of $X_n\setminus S$ and then
add the remaining elements. Thus deleting $S$ leaves the graph connected.
If instead $S=X_n$, every remaining vertex is connected downward to a singleton,
and any two distinct singletons $\{i\}$ and $\{j\}$ are connected through
\(\{i\},\{i,k\},\{k\},\{j,k\},\{j\}\) for some
$k\in X_n\setminus\{i,j\}$. Hence deleting $X_n$ also leaves the graph
connected. Therefore the whole Hasse diagram is one non-trivial block.

Let $K_0:=\{\,\{i\}\mid i\in X_n\,\}$.
If $A,A'\in\Cl(B_n)$ are incompatible, then $I=A\cap A'$ is nonempty. Since
all nonempty subsets of $X_n$ are clusters, the singletons contained in $I$ are
members of $K_0$ and their union is $I$. Hence $K_0$ is a generating set, so
$\mu(B_n)\le n$. Conversely, fix $i\in X_n$ and choose distinct $j,k\ne i$.
The incompatible clusters $\{i,j\}$ and $\{i,k\}$ have intersection $\{i\}$,
so any generating set of $B_n$ contains some $E$ with
$i\in E\subseteq\{i\}$. Hence $E=\{i\}$, and $\mu(B_n)\ge n$.
Thus $\mu(B_n)=n$.

Finally, a vertex $A\subseteq X_n$ has indegree $n-|A|$ in the directed Hasse
diagram. Hence it is a non-root hybrid vertex exactly when
$1\le |A|\le n-2$. The number of such vertices is
\(\sum_{i=1}^{n-2}\binom{n}{i}=2^n-n-2\).
Thus the gap between the Hasse-diagram hybrid count and $\mu(B_n)$ grows
exponentially with $n$.
\end{proof}

\section{Main Results and Consequences}
\label{sec:main-results}

We now state the results before giving their main proofs in
Section~\ref{sec:proofs-main}.  The technical lemmas needed by those proofs are
kept in the appendix.  The first result is the network-to-Hasse lower bound:
any realizing network must contain, in a corresponding network block, enough
non-root hybrid vertices to generate the incompatibility intersections of the
Hasse block.

\begin{proposition}\label{prop:realizing-network-lower-bound}
Let $X$ be a finite nonempty set, let $\mathcal C$ be a cluster system on $X$,
let $H:=\mathcal H[\mathcal C]$, and let $N$ be a rooted network on $X$ with
$C_N=\mathcal C$.  For every non-trivial block $B_H$ of $H$, there exists a
non-trivial block $B_N$ of $N$ such that every cluster in $V(B_H)$ is realized
by some vertex of $B_N$, and $\mu(B_H)\le h(B_N)$.
\end{proposition}

The converse direction is constructive.  Given prescribed generating sets in
the original Hasse blocks, the Hasse diagram can be split without changing the
hardwired cluster system until all remaining non-root hybrid labels in each
current block come from the prescribed set.  The split operation and the
original Hasse-label map used in the statement are defined in
Appendix~\ref{app:splitting}.

\begin{theorem}\label{thm:prescribed-generator-realization}
Let $X$ be a finite nonempty set, let $\mathcal C$ be a cluster system on $X$,
and let $H:=\mathcal H[\mathcal C]$.  For each non-trivial block $B$ of $H$,
fix a finite generating set $K_B\subseteq \Cl(B)$.  Then there exists a rooted
network $N_K$ on $X$, obtained from $H$ by splits and carrying the resulting
original Hasse-label map $\lambda$, with $C_{N_K}=\mathcal C$, such that every
non-trivial block $D$ of $N_K$ has an original non-trivial block $B$ of $H$
with $\lambda(v)\in V(B)$ for every $v\in V(D)$, and every non-root hybrid
vertex of $D$ has original Hasse-label in $K_B$. In particular,
$h(D)\le |K_B|$.
\end{theorem}

Taking minimum generating sets in the preceding theorem and combining it with
the lower bound gives the exact value of the minimum possible level.

\begin{theorem}\label{thm:exact-minimum-level}
Let $X$ be a finite nonempty set, let $\mathcal C$ be a cluster system on $X$,
and let $H:=\mathcal H[\mathcal C]$. Then
\[
\ell(\mathcal C)=
\max\{\,\mu(B)\mid B\text{ is a non-trivial block of }H\,\}.
\]
If $H$ has no non-trivial block, the maximum is interpreted as $0$.
\end{theorem}

Equivalently, the fixed-level decision problem is decided block by block in the
Hasse diagram.

\begin{corollary}\label{cor:level-k-realizability}
Let $X$ be a finite nonempty set, let $\mathcal C$ be a cluster system on $X$,
and let $H:=\mathcal H[\mathcal C]$. Then there exists a rooted level-$k$
network $N$ on $X$ such that $C_N=\mathcal C$ if and only if every non-trivial
block $B$ of $H$ satisfies $\mu(B)\le k$.
\end{corollary}

The proof is organized into two main packages. The lower-bound package shows that
each non-trivial Hasse block is represented inside a single network block and
that the non-root hybrid vertices of that block generate all required
incompatibility intersections. The upper-bound package starts from the Hasse
diagram and repeatedly splits bad hybrid vertices; its key invariants are
cluster preservation, block localization, and termination.

\subsection{Why the Block Parameter is the Correct Threshold}

The point of $\mu(B)$ is not just that it gives another lower bound.  It is the
blockwise threshold for level realizability because it is both forced by every
realization and sufficient for a realization after the prescribed splitting
construction.  This distinguishes it from standard network-side measures of
reticulation complexity, which are defined only after a realizing network has
been chosen.

Let $B_H$ be a non-trivial block of $H$, and let $N$ be a rooted network with
$C_N=\mathcal C$. By Proposition~\ref{prop:hasse-block-realized-in-one-network-block}(b), the block $B_H$ is
represented inside some non-trivial block $B_N$ of $N$. The necessity argument
gives $\mu(B_H)\le h(B_N)$. Thus $\mu(B_H)$ is a lower bound, depending only on
the cluster system, for the number of non-root hybrid vertices required in any
realizing block. Intuitively, each incompatibility-intersection in $B_H$ must be
explained in the network by descendant clusters of non-root hybrid vertices; the
generating set measures how many such pieces are forced.

This inequality need not be an equality for a fixed realization. A realization
may contain redundant hybrid vertices that do not create new hardwired clusters,
and then one can have $\mu(B_H)<h(B_N)$. For instance, in
Figure~\ref{fig:example-mu-one}, the Hasse diagram itself has
$\mu(B_1)=h(B_1)=1$. A non-minimal realization could contain additional
hybrid vertices with the same descendant cluster, leaving the cluster system
unchanged but increasing $h(B_N)$.

The prescribed-generator theorem shows that this lower bound is sharp after
choosing the realization appropriately. Corollary~\ref{cor:level-k-realizability}
states that a level-$k$ realization exists exactly when $\mu(B)\le k$ for every
non-trivial block $B$ of $H$. Thus $\mu(B)$ gives the correct blockwise
threshold, even though arbitrary realizations may contain more hybrid vertices
than necessary.

If one uses the standard reticulation number
$r(N):=\sum_{v\in V(N)}(\indeg_N(v)-1)$, then it measures excess indegree rather
than the number of hybrid vertices. Locally, for a non-trivial block $B_N$, one
may consider
\[
r_{\mathrm{loc}}(B_N):=
\sum_{h\in H^\times(B_N)}(\indeg_N(h)-1).
\]
Since every non-root hybrid vertex has indegree at least two, one has
$h(B_N)\le r_{\mathrm{loc}}(B_N)$, and hence, for a realizing block associated
with $B_H$, $\mu(B_H)\le h(B_N)\le r_{\mathrm{loc}}(B_N)$. In binary settings the
last two quantities coincide locally, while in non-binary settings the
reticulation contribution may be larger.

The level of a network is a bound on $h(B_N)$ over all non-trivial blocks
$B_N$ of that network. Hence level is network-dependent. By contrast, $\mu(B)$
is computed before constructing any network. Theorem~\ref{thm:exact-minimum-level}
converts the network-side minimum-level problem into the Hasse-diagram formula
$\ell(\mathcal C)=\max_B\mu(B)$.

Proposition~\ref{prop:set-cover-formulation} gives a finite set-cover
formulation of $\mu(B)$. We use this only as a finite combinatorial formulation;
we do not claim a polynomial-time algorithm for computing $\mu(B)$, and leave
the complexity classification for future work. It is also
not merely the number of hybrid vertices already present in the Hasse diagram:
Proposition~\ref{prop:boolean-lattice-gap} gives a family with
$\mu(B_n)=n$ but with $2^n-n-2$ non-root hybrid vertices in the Hasse diagram
block.

In summary, $\mu(B)$ turns a network-level requirement into a cluster-system
condition. It detects, directly from $\mathcal H[\mathcal C]$, how much
blockwise reticulation complexity is forced in any level-$k$ realization.

\subsection{Computational consequences}

The set-cover formulation gives a practical way to treat
$\mu(B)$ as a data-side complexity score.  For each Hasse block $B$, the
universe consists of pairs $(I,x)$ where $I$ is an incompatibility intersection
and $x\in I$; a candidate cluster piece $E\in\Cl(B)$ covers $(I,x)$ precisely
when $x\in E\subseteq I$.  Thus $\mu(B)$ can be computed, searched, or
approximated from the cluster system alone, before any network is constructed.

For small or moderately sized blocks, this finite set-cover instance can
be solved exactly by enumeration or integer programming.  For larger blocks,
standard set-cover heuristics can be used to obtain interpretable upper bounds
on the forced local complexity, while the theorem identifies the exact value
whenever the optimum is certified.  Hence the main result is not only an
existence theorem for level-$k$ realizations; it also supplies a route for
turning hardwired cluster data into a blockwise reticulation-complexity
profile.

\section{Proofs of the Main Results}
\label{sec:proofs-main}

We now prove the four results stated in Section~\ref{sec:main-results}.  The
proofs use the lower-bound tools from Appendix~\ref{app:lower-bound} and the
splitting construction, cut-vertex localization, and termination results from
Appendices~\ref{app:splitting}--\ref{app:termination}.  Keeping these proofs in
the main text makes the structure of the argument explicit: the appendix supplies
the technical lemmas, while the body shows how those lemmas combine to give the
minimum-level formula.

\begin{proof}[Proof of Proposition~\ref{prop:realizing-network-lower-bound}]
Let $B_H$ be a non-trivial block of $H$.

By Proposition~\ref{prop:hasse-block-realized-in-one-network-block}(b), there
exists a non-trivial block $B_N$ of $N$ such that every cluster in $V(B_H)$ is
realized by a vertex of $B_N$. Proposition~\ref{prop:network-generators-localize}(b)
then gives $\mu(B_H)\le \mu(B_N)\le h(B_N)$.
\end{proof}

\begin{proof}[Proof of Theorem~\ref{thm:prescribed-generator-realization}]
For each non-trivial block $B$ of $H$, keep the prescribed finite generating
set $K_B\subseteq \Cl(B)$.

Start with the relabeled Hasse diagram described in the initial-network
convention and set $N_0:=H$. Then $N_0$ is a rooted network on leaf set $X$, and
$C_{N_0}=\mathcal C$.

Iteratively do the following. If some original block $B_i$ satisfies
$H_{B_i}^{\mathrm{bad}}(N)\neq\varnothing$, choose a bad hybrid vertex of
lexicographically minimal key; equivalently, choose the smallest index $i$ with
$H_{B_i}^{\mathrm{bad}}(N)\neq\varnothing$ and then a vertex
$u\in H_{B_i}^{\mathrm{bad}}(N)$ of minimum depth in $B_i$. By
Proposition~\ref{prop:admissible-bad-hybrid-splits}(d), $u$ is a valid split
source; split $u$.

By Proposition~\ref{prop:original-cut-never-split}, whenever a vertex $u$ is
split, its original label $\lambda(u)$ is not a cut vertex of the original
Hasse diagram. Hence Proposition~\ref{prop:block-refinement} applies at every
step. By repeated application of Proposition~\ref{prop:splitting-invariants}(a),
every intermediate graph is a rooted network with leaf set $X$ and cluster
system $\mathcal C$. By Proposition~\ref{prop:termination-final-block-control}(a),
the process stops after finitely many steps. Let $N^\ast$ be the resulting
network. Then $C_{N^\ast}=C_{N_0}=\mathcal C$.

It remains to control the non-trivial blocks of $N^\ast$. First, every
non-trivial block of $N^\ast$ is contained in $N^\ast\langle B\rangle$ for some
original non-trivial block $B$ of $H$. This is true for $N_0=H$, since its
non-trivial blocks are exactly the original ones, and the induction step is
Proposition~\ref{prop:block-refinement}(3).

Now let $D$ be any non-trivial block of $N^\ast$, and choose $B$ with
$D\subseteq N^\ast\langle B\rangle$. If $w$ is a non-root hybrid vertex of $D$,
then $\lambda(w)\in V(B)$ and, because the construction has stopped,
$\lambda(w)\in K_B\cup\{\max B\}$. Proposition~\ref{prop:termination-final-block-control}(b)
excludes $\lambda(w)=\max B$, so $\lambda(w)\in K_B$.

By Proposition~\ref{prop:block-refinement}, each element of $K_B$ has a unique
current copy in $N^\ast$. Thus distinct non-root hybrid vertices of $D$ have
distinct labels in $K_B$, and their number is at most $|K_B|$.

Set $N_K:=N^\ast$. The containment $D\subseteq N^\ast\langle B\rangle$ proves
that $\lambda(v)\in V(B)$ for every $v\in V(D)$, and the preceding paragraphs
prove the stated non-root hybrid label control.
Proposition~\ref{prop:splitting-invariants}(a) gives $C_{N_K}=\mathcal C$.
\end{proof}

\begin{proof}[Proof of Theorem~\ref{thm:exact-minimum-level}]
If $H$ has no non-trivial block, then the relabeled Hasse diagram is a rooted
level-$0$ network with cluster system $\mathcal C$, so $\ell(\mathcal C)=0$.
Assume now that $H$ has at least one non-trivial block, and let
$m:=\max\{\,\mu(B)\mid B\text{ is a non-trivial block of }H\,\}$.

First, $m$ is finite. Indeed, the relabeled Hasse diagram itself is a finite
rooted network realizing $\mathcal C$; applying
Proposition~\ref{prop:realizing-network-lower-bound} to this realization shows
that each $\mu(B)$ is bounded by the finite level of $H$.

Let $N$ be any rooted level-$k$ network with $C_N=\mathcal C$. For every
non-trivial block $B$ of $H$, Proposition~\ref{prop:realizing-network-lower-bound}
gives a non-trivial block $B_N$ of $N$ with $\mu(B)\le h(B_N)\le k$.
Hence $m\le k$ for every such realization, and therefore
$m\le \ell(\mathcal C)$.

Conversely, for each non-trivial block $B$ of $H$, choose a minimum generating
set $K_B\subseteq \Cl(B)$ with $|K_B|=\mu(B)$.  Applying
Theorem~\ref{thm:prescribed-generator-realization} to these sets gives a rooted
network $N_K$ with $C_{N_K}=\mathcal C$ such that every non-trivial block $D$ of
$N_K$ satisfies $h(D)\le |K_B|=\mu(B)\le m$ for some original non-trivial
block $B$ of $H$. Thus $N_K$ is level-$m$, so
$\ell(\mathcal C)\le m$.

The two inequalities give the claimed equality.
\end{proof}

\begin{proof}[Proof of Corollary~\ref{cor:level-k-realizability}]
This is immediate from Theorem~\ref{thm:exact-minimum-level}.
\end{proof}

\appendix

\section{Lower Bound: From Realizing Networks Back to Hasse Blocks}
\label{app:lower-bound}

This section collects the tools used in the necessity direction of
Proposition~\ref{prop:realizing-network-lower-bound}.  The aim is to pass from a realizing
network back to the Hasse diagram: incompatibility intersections in a Hasse block are
first located inside a single network block, and then network-side generators
are pulled back to generators in the original Hasse block.  The proof is
organized in three steps.  First, incompatible clusters force their
representatives into one non-trivial network block.  Second, the non-root
hybrid vertices of that block generate the relevant incompatibility
intersections.  Third, the network-side generators are localized back to the
original Hasse block without increasing their number.

\subsection{Compatibility tools inside blocks}

The first package records the incompatibility facts used throughout the lower
bound.  It combines the elementary common-descendant cycle argument with the
standard structural facts from \citet{HellmuthSchallerStadler2023}.

\begin{proposition}
\label{prop:incompatibility-facts-realizing-blocks}
Let $N$ be a rooted network on $X$ with $C_N=\mathcal C$, and put
$H=\mathcal H[\mathcal C]$.
\begin{enumerate}[label=(\alph*)]
\item If $x,y\in V(N)$ are incomparable and have a common descendant leaf, then
$x$ and $y$ lie in a common non-trivial block of $N$.
\item If $a,b\in V(N)$ and $\sigma(a),\sigma(b)$ are incompatible, then $a$
and $b$ lie in the interior of a common non-trivial block of $N$.
\item If $B_H$ is a non-trivial block of $H$ and $V\in B_H^0$, then $V$ has an
incompatible partner in $B_H$.
\item If $B_N$ is a non-trivial block of $N$ and $A,A'\in\Cl(B_N)$ are
incompatible, then there exists $H_{A,A'}\subseteq H^\times(B_N)$ such that
$A\cap A'=\bigcup_{h\in H_{A,A'}}\sigma(h)$.
\end{enumerate}
\end{proposition}

\begin{proof}
For (a), let $\rho$ be the root. Choose directed paths $R_x:\rho\leadsto x$
and $R_y:\rho\leadsto y$. Let $u$ be their last common vertex, and let
$P_x:u\leadsto x$ and $P_y:u\leadsto y$ be the terminal subpaths after $u$.
Then $V(P_x)\cap V(P_y)=\{u\}$. Let $a$ be a common descendant leaf of $x$ and
$y$. Choose directed paths $S_x:x\leadsto a$ and $S_y:y\leadsto a$. Let $v$ be
their first common vertex, and let $Q_x:x\leadsto v$ and $Q_y:y\leadsto v$ be
the initial subpaths up to $v$. Then $V(Q_x)\cap V(Q_y)=\{v\}$.

The four paths $P_x,Q_x,Q_y,P_y$ form a simple undirected cycle. Extra
intersections on the same side would create a directed cycle, while a
cross-intersection would make $x$ and $y$ comparable. Hence $x$ and $y$ lie in
a common non-trivial block.

Part (b) is Lemma~19 of \citet{HellmuthSchallerStadler2023}. For (c), by
Proposition~2 and Theorem~2 of \citet{HellmuthSchallerStadler2023}, $H$ is
semi-regular, and Lemma~28 of that paper gives the required incompatible
partner.

For (d), choose $u,v\in V(B_N)$ with $\sigma(u)=A$ and $\sigma(v)=A'$. By
Lemma~20 of \citet{HellmuthSchallerStadler2023}, if
\[
H_{u,v}:=\{\,h\mid h \text{ is a non-root hybrid vertex of } B_N
\text{ and } h\preceq_N u,v\,\},
\]
then
\[
\sigma(u)\cap\sigma(v)\in
\left\{\sigma(u),\sigma(v),\bigcup_{h\in H_{u,v}}\sigma(h)\right\}.
\]
Since $A$ and $A'$ are incompatible, the first two cases are impossible. By
Lemma~11 of \citet{HellmuthSchallerStadler2023}, $H_{u,v}\subseteq
H^\times(B_N)$. Thus $H_{A,A'}:=H_{u,v}$ has the stated property.
\end{proof}

\subsection{Realizing a Hasse block inside one network block}

We now relate non-trivial blocks of the Hasse diagram to non-trivial blocks of
a realizing network.  The edge-to-block realization lemma is the main local
step: each Hasse edge in a non-trivial Hasse block can be represented inside
one non-trivial block of the realizing network.  The following block-cut-tree
and ear arguments then upgrade this edgewise statement to an entire Hasse
block.

\begin{proposition}\label{prop:edge-to-block-realization}
Let $H=\mathcal H[\mathcal C]$, let $B_H$ be a non-trivial block of $H$, and let $U\to V$ be an edge of $B_H$. Let $N$ be a rooted network with $C_N=\mathcal C$. Then there exists a non-trivial block $B_N$ of $N$ and vertices $u,v\in V(B_N)$ such that $\sigma(u)=U$ and $\sigma(v)=V$.
\end{proposition}

\begin{proof}
Choose $u_0,v_0\in V(N)$ with $\sigma(u_0)=U$ and $\sigma(v_0)=V$. Since
$U\to V$ is a cover edge in the Hasse diagram, $V\subsetneq U$ and there is no
cluster of $\mathcal C$ strictly between $U$ and $V$.

The proof first reduces to a genuine network edge whose endpoint clusters are
$U$ and $V$. If $u_0$ and $v_0$ are incomparable, then every leaf in $V$ is a
common descendant leaf of both vertices, so
Proposition~\ref{prop:incompatibility-facts-realizing-blocks}(a)
already puts them in a common non-trivial block. Thus assume that they are
comparable. The relation $v_0\leadsto u_0$ would imply
$U=\sigma(u_0)\subseteq\sigma(v_0)=V$, impossible; hence $u_0\leadsto v_0$.
Along a directed path
\(u_0=x_0\to x_1\to\cdots\to x_m=v_0\),
the clusters form a descending chain between $U$ and $V$. By the cover
property, every cluster on this path is either $U$ or $V$. Therefore, for some
$j$, there is an edge $u_1\to v_1$ with
$\sigma(u_1)=U$ and $\sigma(v_1)=V$.

It remains to show that the edge $u_1v_1$ lies on an undirected cycle.  In
each case below we build a connected subgraph of $N-u_1v_1$ containing both
$u_1$ and $v_1$. We use the following elementary claim: if $e=xy$ is an edge
of an undirected graph $G$ and a connected subgraph of $G-e$ contains $x$ and
$y$, then $x$ and $y$ lie in a common non-trivial block of $G$. Indeed, the
connected subgraph contains an $x$--$y$ path avoiding $e$, and this path
together with $e$ contains a cycle through $x$ and $y$.

\emph{Case 1: $V$ has an incompatible partner in $B_H$.}
Choose $Z\in V(B_H)$ that is incompatible with $V$, and let $z\in V(N)$ satisfy $\sigma(z)=Z$. By Proposition~\ref{prop:incompatibility-facts-realizing-blocks}(b), $v_1$ and $z$ lie in the interior of a common non-trivial block $B_V$, so there is an undirected path $P_V$ in $B_V$ from $z$ to $v_1$.

\emph{Subcase 1a: $U$ is incompatible with $Z$.}
Again by Proposition~\ref{prop:incompatibility-facts-realizing-blocks}(b), $u_1$ and $z$ lie in the interior of a common non-trivial block $B_U$, so there is an undirected path $P_U$ in $B_U$ from $u_1$ to $z$. If $P_V$ passes through $u_1$, then $u_1,v_1\in B_V$; if $P_U$ passes through $v_1$, then $u_1,v_1\in B_U$. Otherwise $M:=P_U\cup P_V$ is connected, contains $u_1$ and $v_1$, and avoids the edge $u_1v_1$. The claim puts $u_1$ and $v_1$ in a common non-trivial block.

\emph{Subcase 1b: $U$ is compatible with $Z$.}
Since $V\subsetneq U$ and $V$ is incompatible with $Z$, we have $U\cap Z\neq\varnothing$. Thus $U$ and $Z$ are comparable. They cannot satisfy $U\subseteq Z$, since then $V\subseteq Z$, contradicting that $V$ is incompatible with $Z$. Hence $Z\subsetneq U$.

If $u_1$ and $z$ are incomparable, then any leaf in $Z$ is a common descendant leaf of both, so Proposition~\ref{prop:incompatibility-facts-realizing-blocks}(a) yields a non-trivial block $B'$ containing them and an undirected path $P'$ in $B'$ from $u_1$ to $z$. If $P'$ passes through $v_1$, then $u_1,v_1\in B'$; if $P_V$ passes through $u_1$, then $u_1,v_1\in B_V$. Otherwise $M:=P'\cup P_V$ is connected, contains $u_1,v_1$, and avoids $u_1v_1$, so the claim applies.

If $u_1$ and $z$ are comparable, then $z\leadsto u_1$ is impossible because it would imply $U\subseteq Z$. Hence $u_1\leadsto z$; let $R$ be such a path. It cannot pass through $v_1$, since that would force $Z\subseteq V$, again contradicting that $V$ is incompatible with $Z$. If $P_V$ passes through $u_1$, then $u_1,v_1\in B_V$; otherwise $M:=R\cup P_V$ is connected, contains $u_1,v_1$, and avoids $u_1v_1$, so the claim yields the desired block.

\emph{Case 2: $V$ has no incompatible partner in $B_H$.}
By Proposition~\ref{prop:incompatibility-facts-realizing-blocks}(c), $V\notin B_H^0$. Since also $V\neq\max B_H$, the vertex $V$ is terminal in $B_H$. Lemma~\ref{lem:terminal-two-parents} therefore gives another parent $W\neq U$ with $W\to V$ in $B_H$. Then $U$ and $W$ both properly contain $V$ and cannot be comparable, so they are incompatible.

Choose $w_1\in V(N)$ with $\sigma(w_1)=W$. By Proposition~\ref{prop:incompatibility-facts-realizing-blocks}(b), $u_1$ and $w_1$ lie in the interior of a common non-trivial block $B_{U,W}$, so there is an undirected path $P_{u,w}$ in $B_{U,W}$ from $u_1$ to $w_1$.

If $w_1$ and $v_1$ are incomparable, then any leaf in $V$ is a common descendant leaf of both, so Proposition~\ref{prop:incompatibility-facts-realizing-blocks}(a) yields an undirected path $Q$ from $w_1$ to $v_1$ inside a common non-trivial block. If $Q$ passes through $u_1$, then $u_1,v_1$ lie in that block; if $P_{u,w}$ passes through $v_1$, then $u_1,v_1\in B_{U,W}$. Otherwise $M:=P_{u,w}\cup Q$ is connected, contains $u_1,v_1$, and avoids $u_1v_1$, so the claim applies.

If $w_1$ and $v_1$ are comparable, then $v_1\leadsto w_1$ is impossible because it would imply $W\subseteq V$. Hence $w_1\leadsto v_1$; let $R$ be such a path. Along $R$, every cluster lies between $W$ and $V$, hence is either $W$ or $V$ by the cover property. Since $\sigma(u_1)=U\neq W,V$, the path $R$ avoids $u_1$. If $P_{u,w}$ passes through $v_1$, then $u_1,v_1\in B_{U,W}$; otherwise $M:=P_{u,w}\cup R$ is connected, contains $u_1,v_1$, and avoids $u_1v_1$, so the claim applies.

In all cases, $u_1$ and $v_1$ lie in a common non-trivial block of $N$.
\end{proof}

The next result upgrades edgewise realization to blockwise realization.  If
$K\subseteq G$ is a subgraph, a path with endpoints in $K$ and internal
vertices outside $K$ is called a $K$-ear.

\begin{proposition}
\label{prop:hasse-block-realized-in-one-network-block}
Let $N$ be a rooted network with $C_N=\mathcal C$, and let
$H=\mathcal H[\mathcal C]$.
\begin{enumerate}[label=(\alph*)]
\item A non-trivial block of $N$ has no external ear: if a path has two
distinct endpoints in the block and all internal vertices outside it, then the
path has no internal vertices.
\item If $B_H$ is a non-trivial block of $H$, then there exists a non-trivial
block $B_N$ of $N$ such that every $A\in V(B_H)$ has a representative
$v_A\in V(B_N)$ with $\sigma(v_A)=A$.
\end{enumerate}
\end{proposition}

\begin{proof}
For (a), the union of a biconnected block and an external ear is again
biconnected \citep[Proposition~3.1.1]{Diestel2017}. Maximality of the block
therefore forces the whole path to lie in the block.

For (b), we use two claims. First, if $a$ is a cut vertex of a connected graph
$G$, then two blocks of $G$ lie in the same component of the block-cut tree
$T(G)-a$ exactly when their vertex sets, after deleting $a$, lie in the same
component of $G-a$. This follows directly by translating a path in $G-a$ into
a walk in the block-cut tree, and conversely.

Second, let $p=(A_0,A_1,\dots,A_{m-1},A_0)$ be a simple cycle in $B_H$. For
each edge $A_iA_{i+1}$, choose a non-trivial block $D_i$ of $N$ containing
representatives of both endpoint clusters; this is possible by
Proposition~\ref{prop:edge-to-block-realization}. We claim that all $D_i$ are
equal. Let $S$ be the minimal subtree of the block-cut tree of $N$ containing
the block-nodes $D_i$. If $S$ contained a cut-vertex-node $a$, put
$A:=\sigma(a)$. Consecutive blocks $D_i,D_{i+1}$ share representatives of
$A_{i+1}$. If $A_{i+1}\ne A$, these representatives have the same cluster not
equal to $A$ and remain connected in $N-a$: if they are incomparable, use
Proposition~\ref{prop:incompatibility-facts-realizing-blocks}(a) and the
resulting common block remains connected after deleting $a$; if they are
comparable, use the directed path between them, whose vertices all have the
same cluster and hence avoid $a$. The first claim then puts
$D_i$ and $D_{i+1}$ in the same component of $T(N)-a$. If $A_{i+1}=A$, the
same argument is applied around the opposite side of the cycle. Thus all
$D_i$ lie in one component of $T(N)-a$, contradicting that $a\in S$. Hence
$S$ is a single block-node, so $D_0=\cdots=D_{m-1}$.

Now choose a cycle $p_0$ in the non-trivial Hasse block $B_H$. The previous
paragraph gives a non-trivial block $B_N$ of $N$ representing every vertex of
$p_0$. Let $K\subseteq B_H$ be maximal among $2$-connected subgraphs containing
$p_0$ whose vertices all have representatives in this fixed block $B_N$. We
claim that $K=B_H$. If not, either an edge of $B_H$ lies outside $K$, or a
component of $B_H-V(K)$ attaches to $K$ in at least two places; in both cases
there is a $K$-ear $P$ in $B_H$.

Let $a,b$ be the endpoints of $P$, and let $Q$ be a simple $a$--$b$ path in
$K$. The cycle $P\cup Q$ is handled by the second claim: edges of $Q$ already
use the block $B_N$, while edges of $P$ have blocks supplied by
Proposition~\ref{prop:edge-to-block-realization}. Hence every vertex of $P$
also has a representative in $B_N$. Since $K\cup P$ is $2$-connected
\citep[Proposition~3.1.1]{Diestel2017}, this contradicts maximality of $K$.
Therefore $K=B_H$, proving (b).
\end{proof}

\subsection{Pulling network generators back to the Hasse diagram}

The last lower-bound package converts generators in the realizing network back
to generators in the corresponding Hasse block.

\begin{proposition}
\label{prop:network-generators-localize}
The following hold.
\begin{enumerate}[label=(\alph*)]
\item Let $G$ be a connected graph, let $B$ be a block of $G$, and let
$x\in V(G)\setminus V(B)$. If $c$ is the first cut-vertex-node after $B$ on
the path in the block-cut tree from $B$ to the node containing $x$, then
$c\in V(B)$ and every path in $G$ from $V(B)$ to $x$ contains $c$.
\item Let $B_H$ be a non-trivial block of $H=\mathcal H[\mathcal C]$, and let
$B_N$ be a non-trivial block of a realizing network $N$ such that every cluster
in $V(B_H)$ is realized by a vertex of $B_N$. Then
\[
\mu(B_H)\le \mu(B_N)\le h(B_N).
\]
\end{enumerate}
\end{proposition}

\begin{proof}
For (a), the first node after $B$ on the block-cut tree path is a cut vertex
contained in $B$. If a path in $G$ from $V(B)$ to $x$ avoided this cut vertex,
then reading off the blocks and cut vertices traversed by the path would give a
walk in the block-cut tree avoiding $c$, from $B$ to the node containing $x$.
This contradicts the uniqueness of the block-cut tree path.

For (b), first note that $\mu(B_N)\le h(B_N)$. Indeed, let
$K_N^0:=\{\sigma(h)\mid h\in H^\times(B_N)\}\subseteq \Cl(B_N)$.
For every $I=A\cap A'\in\mathcal J(B_N)$,
Proposition~\ref{prop:incompatibility-facts-realizing-blocks}(d) writes
$I$ as a union of clusters $\sigma(h)$ with $h\in H^\times(B_N)$ and
$\sigma(h)\subseteq I$. Hence $K_N^0$ is a generating set of $B_N$, and
$\mu(B_N)\le |K_N^0|\le h(B_N)$.

Let $K_N\subseteq\Cl(B_N)$ be a minimum generating set of $B_N$. Define
\[
S_{B_H}:=\{\,E\in K_N\mid \exists I\in\mathcal J(B_H)\text{ with }E\subseteq I\,\}.
\]
For each $E\in S_{B_H}$ we choose a cluster $\tau(E)\in V(B_H)$ as follows.
If $E\in V(B_H)$, set $\tau(E):=E$. If $E\notin V(B_H)$, choose
$I_0=A_0\cap A_0'\in\mathcal J(B_H)$ with $E\subseteq I_0\subseteq A_0$. Since
$E$ is a cluster of $\mathcal C$, it is a vertex of $H$, and there is a
directed path in $H$ from $A_0$ to $E$. Let $\tau(E)$ be the first cut vertex
after the block-node $B_H$ on the block-cut tree path from $B_H$ to the node
containing $E$.

By part (a), $\tau(E)\in V(B_H)$ and every path in $H$ from $V(B_H)$ to $E$
passes through $\tau(E)$. Therefore, whenever $E\subseteq I=A\cap A'$ with
$A,A'\in V(B_H)$ incompatible, the two downward paths from $A$ and $A'$ to
$E$ both pass through $\tau(E)$, and hence $E\subseteq \tau(E)\subseteq I$.

Set $K_{B_H}:=\{\tau(E)\mid E\in S_{B_H}\}$. We show that $K_{B_H}$ generates
$B_H$. Let $I=A\cap A'\in\mathcal J(B_H)$. Since $A,A'$ are realized in
$B_N$, we also have $I\in\mathcal J(B_N)$. Because $K_N$ generates $I$, each
element of $I$ lies in some $E\in K_N$ with $E\subseteq I$. Such an $E$ lies in
$S_{B_H}$, and the previous paragraph gives $E\subseteq\tau(E)\subseteq I$.
Thus \(I=\bigcup\{\,D\in K_{B_H}\mid D\subseteq I\,\}\).
Consequently $K_{B_H}$ is a generating set and
$\mu(B_H)\le |K_{B_H}|\le |K_N|=\mu(B_N)\le h(B_N)$.
\end{proof}

\section{Upper Bound I: The Splitting Construction}
\label{app:splitting}

This section sets up the splitting machinery used in the sufficiency direction.
The construction starts from the Hasse diagram and repeatedly splits selected
hybrid vertices while preserving the hardwired cluster system.  The main local
tasks are to define which hybrid vertices are still bad, to prove that a split
does not change the represented clusters, and to control where new bad hybrid
vertices can appear after a split.

\subsection{Setting up the splitting construction}

\paragraph{Running convention.}
Every vertex in a current network carries an \emph{original Hasse-label}
$\lambda(v)\in V(H)$.  When a vertex $u$ is split, all new copies inherit the
same label: $\lambda(u_i)=\lambda(u)$.  A vertex with original label $x$ is
called a \emph{current copy} of $x$.

Whenever the construction starts from $N_0=\mathcal H[\mathcal C]$, we identify each inclusion-minimal vertex $\{x\}$ with the taxon $x$; equivalently, we relabel each leaf $\{x\}$ by $x$ while keeping its original Hasse-label $\{x\}$. This does not change descendant clusters, so the resulting rooted network still has cluster system $\mathcal C$, and we continue to denote it by $\mathcal H[\mathcal C]$.

Let $N$ be a rooted network, and let $u$ be a non-root, non-leaf vertex of $N$
with parents $p_1,\dots,p_r$ and children $w_1,\dots,w_t$. The \emph{vertex
split} of $u$ is the directed graph obtained from $N$ by deleting $u$, adding
new vertices $u_1,\dots,u_r$, and inserting the edges $p_i\to u_i$ and
$u_i\to w_j$ for all $i,j$. For brevity we call this operation a \emph{split}.
Starting from the initial network $N_0$ above, any digraph obtained after
finitely many such splits is called a \emph{current network};
Proposition~\ref{prop:splitting-invariants}(a) shows that every current network
is again a rooted network with the same cluster system. For an original
non-trivial block $B$ of $H=\mathcal H[\mathcal C]$ and a vertex $x\in V(B)$,
let $d_B(x)$ be the length of a longest directed path in $B$ from $\max B$ to
$x$; this is the \emph{depth} of $x$ in the original block $B$.

\begin{definition}
Let $x\in V(\mathcal H[\mathcal C])$. If there exists an original non-trivial block $B$ of $\mathcal H[\mathcal C]$ such that $x\in V(B)\setminus\{\max B\}$, then $B$ is the \emph{owner block} of $x$, denoted by $\operatorname{own}(x)$.
\end{definition}

\begin{definition}
Fix, for every non-trivial block $B$ of $\mathcal H[\mathcal C]$, a generating set
$K_B\subseteq \Cl(B)$. Let $N$ be a current network, and let $\lambda(v)$ denote
the original Hasse-label of a vertex $v\in V(N)$. For a fixed original
non-trivial block $B$, define
\[
H_B^{\mathrm{bad}}(N):=\{\,v\in V(N)\mid \lambda(v)\in \Cl(B)\setminus (K_B\cup\{\max B\}),\ \indeg(v)\ge 2\,\}.
\]
Vertices in $H_B^{\mathrm{bad}}(N)$ are called bad hybrid vertices relative to $B$.
\end{definition}

\begin{proposition}\label{prop:admissible-bad-hybrid-splits}
The following bookkeeping facts hold throughout the construction.
\begin{enumerate}[label=(\alph*)]
\item If an original vertex $x$ of $\mathcal H[\mathcal C]$ is a non-root vertex of some original non-trivial block, then $\operatorname{own}(x)$ is well defined and unique.
\item If $B$ is a non-trivial block of $\mathcal H[\mathcal C]$ and $\{x\}\in V(B)$, then every generating set $K_B$ of $B$ contains $\{x\}$.
\item In every current network, a vertex $v$ is a leaf if and only if $\lambda(v)$ is a singleton cluster.
\item If $u\in H_B^{\mathrm{bad}}(N)$, then $u$ is a non-root, non-leaf vertex and can be split.
\end{enumerate}
\end{proposition}

\begin{proof}
For (a), if $x$ is not a cut vertex of $\mathcal H[\mathcal C]$, then it lies in a unique block. If $x$ is a cut vertex and $B$ is a non-trivial block with $x\in V(B)\setminus\{\max B\}$, then $B$ cannot be a child block of the cut-vertex-node $x$ in the rooted block-cut tree; otherwise every path from the root to a vertex of $B$ would enter $B$ through $x$, forcing $x=\max B$. Thus $B$ is the unique parent block of $x$.

For (b), the singleton $\{x\}$ is terminal in $B$. By Lemma~\ref{lem:terminal-two-parents}, it has distinct parents $U,W\in V(B)$ with $U\to\{x\}$ and $W\to\{x\}$. They are incomparable, hence incompatible. Put $I:=U\cap W$. Since $K_B$ generates $I$, there is $E\in K_B$ with $x\in E\subseteq I$. Then $\{x\}\subseteq E\subseteq U$. The edge $U\to\{x\}$ is a cover edge, so $E=U$ or $E=\{x\}$; but $E\subseteq I=U\cap W\subsetneq U$. Hence $E=\{x\}$.

For (c), use induction on the number of splits. Initially the leaves are exactly the singleton clusters. Suppose $N'$ is obtained from $N$ by splitting a bad vertex $u\in H_B^{\mathrm{bad}}(N)$. By (b), $\lambda(u)$ is not a singleton; by induction, $u$ is not a leaf. All old vertices other than $u$ keep their leaf status, and the new copies of $u$ have the same non-empty child set as $u$ and the same non-singleton label. Thus the equivalence is preserved.

For (d), a bad vertex has indegree at least two, so it is not the root. Also $\lambda(u)\in\Cl(B)\setminus(K_B\cup\{\max B\})$, and (b) shows that $\lambda(u)$ is not a singleton. By (c), $u$ is not a leaf.
\end{proof}

Every vertex $v\in H_B^{\mathrm{bad}}(N)$ satisfies $\lambda(v)\in V(B)\setminus\{\max B\}$, so Proposition~\ref{prop:admissible-bad-hybrid-splits}(a) gives $\operatorname{own}(\lambda(v))=B$. In particular, a current vertex belongs to at most one set $H_B^{\mathrm{bad}}(N)$, and for such a vertex the depth $d_B(\lambda(v))$ is well defined.

The next result records the structural invariants of the split operation. It
will be used repeatedly without reproving the same edge-projection
bookkeeping.

\begin{proposition}\label{prop:splitting-invariants}
Let $N'$ be obtained from a current network $N$ by splitting a non-root,
non-leaf vertex $u$.
\begin{enumerate}[label=(\alph*)]
\item $N'$ is again a rooted network, with the same root, leaf set, and cluster system as $N$.
\item If $a\to b$ is an edge of any current network, then $\lambda(a)\to\lambda(b)$ is an edge of the original Hasse diagram $\mathcal H[\mathcal C]$.
\item Let $v$ be any original vertex of $\mathcal H[\mathcal C]$, and let $\operatorname{Par}_{\mathcal H}(v)$ denote the set of original parents of $v$. In every current network, every parent of every current copy of $v$ has original label in $\operatorname{Par}_{\mathcal H}(v)$.
\end{enumerate}
\end{proposition}

\begin{proof}
For (a), write the parents and children of $u$ as $p_1,\dots,p_r$ and
$w_1,\dots,w_t$. Any directed cycle in $N'$ would contain a new vertex $u_i$;
replacing each segment $p_i\to u_i\to w_j$ by $p_i\to u\to w_j$ gives a
directed cycle in $N$, impossible. The old root is not deleted, no new
indegree-$0$ vertex is created, and every old or new vertex remains reachable
from the root by the same segment replacement. The leaf set is unchanged
because $u$ is not a leaf, each new $u_i$ has children, and every parent $p_i$
keeps an outgoing edge. Finally, descendants below each $u_i$ are exactly those
below $u$, and paths through $u$ in $N$ correspond to paths through the copies
of $u$ in $N'$. Hence all vertex clusters are preserved and $C_{N'}=C_N$.

For (b), use induction on the number of splits. The claim is immediate for
$N_0=\mathcal H[\mathcal C]$. If $a\to b$ is a new edge after splitting $u$,
then it is either $p_i\to u_i$ or $u_i\to w_j$, and it projects to the old
current edge $p_i\to u$ or $u\to w_j$. The induction hypothesis gives the
corresponding original Hasse edge.

Part (c) follows immediately from (b): if $p$ is a parent of a current copy
$x$ of $v$, then $\lambda(p)\to\lambda(x)=v$ is an original Hasse edge.
\end{proof}

\subsection{Local behavior of bad hybrid vertices under splitting}

Once a bad hybrid vertex is split, the construction needs a monotonicity
principle: new bad hybrid vertices may appear, but only lower in the same
original block or in descendant blocks of the original block-cut tree.

\begin{definition}
Let $N$ be a current network and let $B$ be an original non-trivial block of
$\mathcal H[\mathcal C]$. Define
\[
N\langle B\rangle:=N\bigl[\{\,v\in V(N)\mid \lambda(v)\in V(B)\,\}\bigr].
\]
\end{definition}

\begin{lemma}\label{lem:new-bad-locality}
Let $u\in H_B^{\mathrm{bad}}(N)$, where $B$ is an original non-trivial block,
and let $N'$ be obtained from $N$ by splitting $u$.
\begin{enumerate}[label=(\alph*)]
\item Every new bad hybrid vertex of $N'$ relative to $B$ is a former child of $u$ and has label of larger $B$-depth.
\item Every new bad hybrid vertex of $N'$ has owner block equal to $B$ or to a descendant non-trivial block of $B$ in the rooted block-cut tree.
\item If $v$ is a current vertex with $\lambda(v)\in V(B)$ and $\lambda(v)$ is not an original cut vertex, then every parent and every child of $v$ has original label in $V(B)$.
\end{enumerate}
\end{lemma}

\begin{proof}
The split changes indegrees only at $u$, the new copies of $u$, and the former
children of $u$; each new copy has indegree $1$. Hence any new bad hybrid
vertex is a former child of $u$.

If such a vertex $v$ is bad relative to $B$, then
$\lambda(u),\lambda(v)\in V(B)\setminus\{\max B\}$. Since $u\to v$ was an edge
of $N$, Proposition~\ref{prop:splitting-invariants}(b) gives the original
Hasse edge $\lambda(u)\to\lambda(v)$. Both endpoints lie in $V(B)$, so this
edge belongs to $B$, and any directed path in $B$ from $\max B$ to
$\lambda(u)$ extends to one ending at $\lambda(v)$. Thus
$d_B(\lambda(v))>d_B(\lambda(u))$, proving (a).

For (b), let $v$ be any new bad hybrid vertex. The same edge
$\lambda(u)\to\lambda(v)$ exists in the original Hasse diagram. Since
$u\in H_B^{\mathrm{bad}}(N)$, Proposition~\ref{prop:admissible-bad-hybrid-splits}(a)
gives $\operatorname{own}(\lambda(u))=B$. If $\lambda(v)\in V(B)$, then
$\lambda(v)\subsetneq\lambda(u)\subsetneq\max B$, so
$\operatorname{own}(\lambda(v))=B$. Otherwise the edge leaves $B$ through the
cut vertex $\lambda(u)$. The block containing that edge is a child block below
$\lambda(u)$, hence a descendant of $B$, and the owner block of $\lambda(v)$
lies in the same descendant subtree.

For (c), induct on the number of splits. Initially, a non-cut vertex of the
Hasse diagram lies in a unique block, so all its incident Hasse edges lie in
that block. During a split, only edges incident with the split vertex change.
If the current copy under consideration is unaffected, there is nothing to
prove; if it is a new copy of the split vertex or an old neighbor of that
vertex, the claim follows from the induction hypothesis applied to the
corresponding old incident edges.
\end{proof}

\section{Upper Bound II: Cut Vertices and Block Localization}
\label{app:localization}

This section proves that the splitting construction does not destroy the
original cut-vertex decomposition of the Hasse diagram.  These results keep
splits localized to one original block and supply the block-control invariant
used in the prescribed-generator realization theorem.  The argument first
rules out splitting original cut vertices, then shows that generator labels are
never split, and finally proves that every affected current block remains
inside the support of a single original Hasse block.

\subsection{Cut vertices are never split}

The first two propositions dispose of hybrid cut vertices.  A hybrid cut
vertex is forced into the generating set of its parent block, and therefore it
can never be selected as a bad hybrid vertex.

\begin{proposition}\label{prop:shared-in-K}
Let $Z$ be a hybrid cut vertex of the original Hasse diagram
$\mathcal H[\mathcal C]$, and put $B_H^\uparrow:=\operatorname{own}(Z)$.
Then $Z\in K_{B_H^\uparrow}$.
\end{proposition}

\begin{proof}
By Lemmas~9--11 of \citet{HellmuthSchallerStadler2023}, a hybrid vertex is a
non-root vertex of exactly one non-trivial block, so $\operatorname{own}(Z)$ is
defined and is the unique parent block of the cut-vertex-node $Z$ in the rooted
block-cut tree.

Choose two distinct direct parents $p_1,p_2$ of $Z$ in
$\mathcal H[\mathcal C]$. Since $Z$ is a non-root vertex of
$B_H^\uparrow$, every direct parent of $Z$ lies in $B_H^\uparrow$. Thus
$p_1,p_2\in V(B_H^\uparrow)$. These parents are incomparable, because otherwise
one of the edges $p_1\to Z$ or $p_2\to Z$ would not be a cover edge. Hence they
are incompatible. Put $P:=p_1\cap p_2$.
Then $P\in\mathcal J(B_H^\uparrow)$ and $Z\subseteq P$.

Assume for contradiction that $Z\notin K_{B_H^\uparrow}$. Since $Z$ is a cut
vertex, there is another block $B_H^\downarrow$ containing $Z$. Every such
block is a child block of the cut-vertex-node $Z$, so it contains vertices
strictly below $Z$. Choose a child $a'$ of $Z$ with
$a'\in V(B_H^\downarrow)$, and then choose a leaf $a$ below $a'$. This gives a
directed path \(P_1:\ Z\to a'\leadsto a\).

Since $a\in Z\subseteq P$ and $P$ is generated by $K_{B_H^\uparrow}$, there is
some $k\in K_{B_H^\uparrow}$ such that $a\in k\subseteq P$.
As $Z\notin K_{B_H^\uparrow}$, we have $k\neq Z$.

We claim that $Z\subsetneq k$ is impossible. Indeed, if $Z\subsetneq k$, then
\[
Z\subsetneq k\subseteq p_1
\qquad\text{and}\qquad
Z\subsetneq k\subseteq p_2,
\]
contradicting that $p_1\to Z$ and $p_2\to Z$ are cover edges. Hence either
$k\subsetneq Z$ or $k$ is incompatible with $Z$.

In either case there is a directed path \(P_2:\ k\leadsto a\) that avoids
$Z$: if $k\subsetneq Z$, then $Z$ lies above $k$; if $k$ is incompatible with
$Z$ and $P_2$ passed through $Z$, then $Z$ would lie below $k$, implying
$Z\subseteq k$, contrary to the assumed incompatibility.

Next we show that $P_1$ does not re-enter $B_H^\uparrow$ after leaving $Z$.
Otherwise let $c\neq Z$ be the first later vertex of $P_1$ that lies in
$B_H^\uparrow$. Then the segment $P_1[Z,c]$ has both endpoints in
$B_H^\uparrow$, has all internal vertices outside $B_H^\uparrow$, and contains
the vertex $a'$. This is an external ear for $B_H^\uparrow$, contradicting
Proposition~\ref{prop:hasse-block-realized-in-one-network-block}(a).

Let $z$ be the last vertex of $P_2$ lying in $B_H^\uparrow$; this vertex exists
because $k\in B_H^\uparrow$. Let $y$ be the first common vertex of
$P_2[z,a]$ and $P_1$. Then \(P_1[Z,y]\cap P_2[z,y]=\{y\}\).
Indeed, $Z\notin V(P_2)$, the vertex $z$ is not on $P_1$ because $P_1$ does
not re-enter $B_H^\uparrow$, and any earlier common vertex would contradict the
choice of $y$.

Hence the union of $P_1[Z,y]$ and the reverse of $P_2[z,y]$ is a simple
undirected path $R$ from $Z$ to $z$. All internal vertices of $R$ lie outside
$B_H^\uparrow$: after leaving $Z$, the path $P_1$ stays outside
$B_H^\uparrow$, and after the last vertex $z$ in $B_H^\uparrow$, the path
$P_2$ also stays outside. Moreover, $R$ has an internal vertex outside
$B_H^\uparrow$, since $a'$ lies on $P_1[Z,y]$. Thus $R$ is an external ear for
$B_H^\uparrow$, again contradicting
Proposition~\ref{prop:hasse-block-realized-in-one-network-block}(a). Therefore
$Z\in K_{B_H^\uparrow}$.
\end{proof}

\begin{proposition}\label{prop:no-shared-bad-split}
Let $Z$ be a hybrid cut vertex of the original Hasse diagram $\mathcal H[\mathcal C]$.
Then $Z$ is never split during the construction.
\end{proposition}

\begin{proof}
Put $B_H^\uparrow:=\operatorname{own}(Z)$. Proposition~\ref{prop:shared-in-K} gives
$Z\in K_{B_H^\uparrow}$. In every other non-trivial block containing $Z$,
Lemma~9 of \citet{HellmuthSchallerStadler2023} forces $Z$ to be the block root;
otherwise $Z$ would be a non-root vertex of two non-trivial blocks. Thus for no
original non-trivial block $B$ do we have
$Z\in\Cl(B)\setminus(K_B\cup\{\max B\})$. A current copy of $Z$ can therefore
never be a bad hybrid vertex, and the construction never splits it.
\end{proof}

The next proposition is the main cut-vertex statement.  The proof contains the
two local claims that were previously isolated: a source-chain claim and an
exclusive element claim.

\begin{proposition}[Original cut vertices are never split]\label{prop:original-cut-never-split}
No original cut vertex of $\mathcal H[\mathcal C]$ is ever split during the construction.
\end{proposition}

\begin{proof}
The proof first reduces a hypothetical first split cut vertex to a hybrid
source chain, and then uses one exclusive element of the cut vertex to obtain
a cover-edge contradiction.

Assume, for contradiction, that some original cut vertex is split. By
Proposition~\ref{prop:no-shared-bad-split}, hybrid cut vertices are never split, so
the first original cut vertex that is split is non-hybrid. Call it $Z$.

Let $B_H^\uparrow$ be the unique original non-trivial block with
$Z\in V(B_H^\uparrow)\setminus\{\max B_H^\uparrow\}$.
Since $Z$ is non-hybrid, it has a unique original parent $p$ in this block.

\emph{Claim 1.} There is an integer $r\ge0$ and a directed path
$c_0\to c_1\to\cdots\to c_r\to Z$ inside $B_H^\uparrow$ such that $c_0$ is hybrid,
$c_1,\dots,c_r,Z$ are non-hybrid, and every $c_j$ is split before $Z$ while
lying outside $K_{B_H^\uparrow}\cup\{\max B_H^\uparrow\}$.

To prove the claim, first note that unique parents persist for non-hybrid
labels: if $v$ is an original non-hybrid vertex with unique original parent
$q$, then every parent of every current copy of $v$ is a current copy of $q$.
This follows by induction on the number of splits, since a split can affect
parents of copies of $v$ only when the split label is $q$ or $v$. Consequently,
if a current copy of $v$ is split, then some current copy of $q$ was split
earlier: the split copy of $v$ has at least two parent copies of $q$, and
additional copies of $q$ can only be created by splitting a copy of $q$.

Since $Z$ is split, a current copy of its unique parent $p$ was split earlier.
Starting with $d_0:=p$, repeatedly replace a non-hybrid $d_j$ by its unique
original parent $d_{j+1}$; the preceding paragraph shows that a current copy
of $d_{j+1}$ was also split before $Z$. The process stops because each step
moves upward in the finite Hasse diagram. Each constructed $d_j$ lies in
$V(B_H^\uparrow)$, is split before $Z$, and does not lie in
$K_{B_H^\uparrow}\cup\{\max B_H^\uparrow\}$: for $d_0=p$ this follows because
$p$ is not a cut vertex before the first split cut vertex, and the induction
step uses Lemma~\ref{lem:new-bad-locality}(c) to keep the unique parent inside
$B_H^\uparrow$. The terminal vertex must be hybrid, otherwise its unique
parent would also have been split before $Z$. Reversing the chain proves the
claim.

\emph{Claim 2.} There is an element $x\in Z$ such that for every
$C\in V(B_H^\uparrow)$, \(x\in C\) implies \(Z\subseteq C\).

Since $B_H^\uparrow$ is non-trivial, $B_H^\uparrow-Z$ is connected. Hence all
vertices of $V(B_H^\uparrow)\setminus\{Z\}$ lie in a single component $U$ of
$\mathcal H[\mathcal C]-Z$. Because $Z$ is a cut vertex, another component
$D\neq U$ exists. The unique parent of $Z$ lies in $B_H^\uparrow$, hence in
$U$, so every neighbor of $Z$ outside $U$ is a child of $Z$. Choose such a
child $a'\in D$ and take $x\in a'$. If some $C\in V(B_H^\uparrow)$ contains
$x$ but not $Z$, then $C\in U$, and the downward paths from $a'$ and from $C$
to $\{x\}$ both avoid $Z$. This connects $a'$ to $C$ in
$\mathcal H[\mathcal C]-Z$, contradicting $a'\in D$ and $C\in U$. Thus the
claim holds.

By Claim 1, $c_0\notin K_{B_H^\uparrow}$. Choose $x\in Z$ as in Claim 2.
Let $A,A'$ be two distinct direct parents of the hybrid vertex $c_0$. The
vertex $c_0$ is not a cut vertex: it is split before $Z$, and $Z$ was chosen as
the first split original cut vertex. Hence all incident Hasse edges of $c_0$
belong to the same original block $B_H^\uparrow$, so
$A,A'\in V(B_H^\uparrow)$. The parents $A$ and $A'$ are incomparable, and
therefore incompatible. Put $I:=A\cap A'$.
Since $x\in Z\subseteq c_0\subseteq I$ and $K_{B_H^\uparrow}$ generates $I$,
there exists $k\in K_{B_H^\uparrow}$ with
$x\in k\subseteq I$.

We first show that $c_0\not\subseteq k$. If $c_0\subseteq k$, then
$c_0\neq k$ because $c_0\notin K_{B_H^\uparrow}$. Hence
\(c_0\subsetneq k\subseteq I\subsetneq A\),
where the last containment is proper because $A$ and $A'$ are incompatible.
This contradicts the fact that $A\to c_0$ is a cover edge of the Hasse diagram.

On the other hand, the choice of $x$ forces $Z\subseteq k$, because
$x\in k$ and $k\in V(B_H^\uparrow)$. Moreover $Z\notin K_{B_H^\uparrow}$: a
current copy of $Z$ can be bad only relative to its owner block
$B_H^\uparrow$, and the assumed split of $Z$ requires its label to be outside
$K_{B_H^\uparrow}\cup\{\max B_H^\uparrow\}$. Thus $Z\subsetneq k$.

Consider any directed path in the Hasse diagram from $k$ down to $Z$. Since
$Z,c_r,\dots,c_1$ are non-hybrid, they each have a unique direct parent. Thus
the final segment of this path must pass successively through
\(c_0,c_1,\dots,c_r,Z\).
Consequently $c_0\subseteq k$, contradicting the conclusion of the previous
paragraph. Therefore no original non-hybrid cut vertex is ever split. Together
with Proposition~\ref{prop:no-shared-bad-split}, this proves that no original cut
vertex of $\mathcal H[\mathcal C]$ is ever split.
\end{proof}

The last proposition packages the localization invariants used after cut
vertices have been fixed.  Its proof contains the generator-label and
cut-separation claims needed for the final block-control argument.

\begin{proposition}[Block refinement under splitting]\label{prop:block-refinement}
Throughout the construction, if $B$ is an original non-trivial block and
$x\in K_B$, then no current copy of $x$ is ever split; hence $x$ has exactly
one current copy at every stage.

Moreover, let $N'$ be obtained from a current network $N$ by splitting a bad
hybrid vertex $u\in H_B^{\mathrm{bad}}(N)$, where $B$ is an original
non-trivial block of $\mathcal H[\mathcal C]$, and assume that $\lambda(u)$ is
not a cut vertex of the original Hasse diagram.

Then:
\begin{enumerate}[label=(\arabic*)]
\item every old edge of $N$ whose incidence changes during the split lies in $N\langle B\rangle$;

\item every non-trivial block of $N'$ that contains either one of the new copies
$u_1,\dots,u_r$ or an edge whose incidence changed during the split is contained in the graph obtained
from $N\langle B\rangle$ by replacing $u$ with its copies $u_1,\dots,u_r$;

\item in particular, if every non-trivial block of $N$ is contained in $N\langle D\rangle$ for some
original non-trivial block $D$, then the same is true for every non-trivial block of $N'$.
\end{enumerate}
\end{proposition}

\begin{proof}
First, generator labels are fixed. If $x$ is an original cut vertex, then no
current copy of $x$ is ever split by Proposition~\ref{prop:original-cut-never-split}.
If $x$ is not a cut vertex, then $x$ lies in the unique original block $B$ and
in no other original non-trivial block. Since $x\in K_B$, a current copy of
$x$ cannot be bad relative to $B$ or to any other original block. Thus no copy
of $x$ is ever chosen for splitting, and because copies of a label are created
only by splitting that label, $x$ has exactly one current copy throughout.

We also record the cut-separation invariant. Let $Z$ be an original cut vertex
and let $\Gamma_1,\dots,\Gamma_m$ be the connected components of
$\mathcal H[\mathcal C]-Z$. Since $Z$ is never split by
Proposition~\ref{prop:original-cut-never-split}, every current network has a
unique current copy of $Z$. For each current network $M$, set
\[
W_j(M):=\{\,v\in V(M)\setminus\{Z\}\mid \lambda(v)\in V(\Gamma_j)\,\}.
\]
Induction on the number of splits shows that $M-Z$ contains no edge between
$W_a(M)$ and $W_b(M)$ for $a\neq b$: initially this is the component
decomposition of $\mathcal H[\mathcal C]-Z$, and a later split only replaces
edges incident with one vertex by edges through copies with the same label.
The induction hypothesis keeps the split vertex and all of its non-$Z$
neighbors in the same set $W_j(M)$.
Thus the separation induced by every original cut vertex persists in every
current network.

By Lemma~\ref{lem:new-bad-locality}(c), every parent and every child of $u$ in the current network has
original label in $V(B)$. Hence every old edge of $N$ whose incidence changes during the split of
$u$ lies in $N\langle B\rangle$. This proves (1).

Let $S$ be obtained from $N\langle B\rangle$ by replacing $u$ with its copies.  To prove (2), let
$D'$ be a non-trivial block of $N'$ containing a new copy, or equivalently an affected edge, and
fix such a copy $u_i\in V(D')$.  If $D'$ contained a vertex $y\notin S$, then
$\lambda(y)\notin V(B)$.  Let $Z$ be the first cut vertex after $B$ on the block-cut tree path from
$B$ to the node containing $\lambda(y)$.  By Proposition~\ref{prop:network-generators-localize}(a), every Hasse
path from $V(B)$ to $\lambda(y)$ passes through $Z$.  Since $\lambda(u)$ is not a cut vertex,
$\lambda(u_i)=\lambda(u)$ and $\lambda(y)$ lie in different components of
$\mathcal H[\mathcal C]-Z$.

The cut-separation invariant above preserves this separation in $N'$.  Thus $u_i$ and $y$ lie in different components of $N'-Z$.  But a non-trivial
block containing both vertices cannot be separated by deleting a single vertex.  This contradiction
shows $D'\subseteq S$, proving (2).

For (3), affected non-trivial blocks are handled by (2).  Any other non-trivial block of $N'$
uses only old vertices and old edges with unchanged incidence, hence is contained in a non-trivial
block of $N$ and therefore, by the induction hypothesis in (3), in some old support
$N\langle D\rangle\subseteq N'\langle D\rangle$.
\end{proof}

\section{Upper Bound III: Termination and Final Block Control}
\label{app:termination}

It remains to prove that the splitting construction stops and that the final
blocks have the label control needed for the upper bound.  This section isolates
the two ingredients used in the final counting argument: the root label of an
original block cannot become a non-root hybrid inside that block, and a
lexicographic counting vector strictly decreases whenever the prescribed split
rule is applied.

\subsection{Termination and final block control}

Let $B_1,\dots,B_s$ be the original non-trivial blocks of
$\mathcal H[\mathcal C]$, enumerated in a topological order of the original
block-cut tree, that is, whenever $B_j$ is a descendant of $B_i$, one has
$i<j$.  For every original vertex $x\in V(\mathcal H[\mathcal C])$, let
$\nu(x)$ denote the number of directed paths from the root $X$ of
$\mathcal H[\mathcal C]$ to $x$.

Since original cut vertices are never split by Proposition~\ref{prop:original-cut-never-split},
every bad hybrid vertex has an original label that is a non-root vertex of a
unique original non-trivial block.  For a bad hybrid vertex $v$, define its
\emph{key} by
\[ 
\kappa(v):=\bigl(i,d_{B_i}(\lambda(v))\bigr),
\qquad\text{where } \operatorname{own}(\lambda(v))=B_i.
\]
For each pair $(i,\ell)$ with $1\le i\le s$ and $0\le\ell\le h_i$, where
$h_i:=\max\{\,d_{B_i}(x):x\in V(B_i)\,\}$,
define
\[
a_{i,\ell}(N):=
\bigl|\{\,v\in H_{B_i}^{\mathrm{bad}}(N)\mid d_{B_i}(\lambda(v))=\ell\,\}\bigr|.
\]
Let $\Theta(N)$ be the vector obtained by listing all coordinates
$a_{i,\ell}(N)$ in lexicographic order of the keys $(i,\ell)$.

\begin{proposition}[Termination and final block control]\label{prop:termination-final-block-control}
If at each step one chooses a bad hybrid vertex of lexicographically minimal
key and splits it, then:
\begin{enumerate}[label=(\alph*)]
\item the splitting process terminates after finitely many steps;
\item in every current network obtained during the process, if $B$ is an original non-trivial block and $D$ is a current non-trivial block contained in $N\langle B\rangle$, then no current copy of $\max B$ is a non-root vertex of $D$.
\end{enumerate}
\end{proposition}

\begin{proof}
We first record two local claims needed for final block control, then prove
termination by a lexicographically decreasing bounded vector.

\emph{Claim 1.} No current copy of $\max B$ is a non-root vertex of a current
non-trivial block $D\subseteq N\langle B\rangle$.

Suppose otherwise. Then there is a directed path in $D$ from $\max D$ to a
current copy $x$ of $\max B$ of positive length. Let $y$ be the predecessor of
$x$ on such a path. Since $D\subseteq N\langle B\rangle$, we have
$\lambda(y)\in V(B)$. On the other hand,
Proposition~\ref{prop:splitting-invariants}(c) implies that
$\lambda(y)$ is an original parent of $\max B$. Since $\max B$ is the unique
maximal vertex of the original block $B$, no original parent of $\max B$ lies
in $V(B)$, a contradiction.

\emph{Claim 2.} A current vertex has at most one child of each original label.

Induct on the number of splits. Initially each original label occurs once.
When a vertex $u$ is split, an old parent of $u$ replaces its single edge to
$u$ by a single edge to the corresponding copy of $u$, and each new copy of
$u$ receives exactly the old child set of $u$. All other outgoing edge sets
are unchanged, so the induction hypothesis preserves the bound.

Let $N'$ be obtained from $N$ by one step of the procedure, so we split a bad hybrid vertex
$u\in H_{B_i}^{\mathrm{bad}}(N)$ of lexicographically minimal key
$\kappa(u)=(i,d)$.

\emph{Lexicographic descent.}
Only former children of $u$ can become new bad hybrid vertices. By
Lemma~\ref{lem:new-bad-locality}(b), such new bad vertices have owner block either
$B_i$ or a descendant of $B_i$ in the original block-cut tree. Since the blocks
$B_1,\dots,B_s$ are topologically ordered, descendant blocks have larger
indices. Within the same block $B_i$, Lemma~\ref{lem:new-bad-locality}(a) puts every
new bad hybrid vertex at depth strictly larger than $d$. Hence no coordinate of
$\Theta$ preceding $(i,d)$ changes, and no new bad hybrid vertex of key
$(i,d)$ is created. The split removes the chosen bad vertex $u$, so
\(a_{i,d}(N')=a_{i,d}(N)-1\), while only later coordinates may change.
Therefore \(\Theta(N')<_{\mathrm{lex}}\Theta(N)\).

\emph{Uniform boundedness.}
It remains to show that the coordinates of $\Theta(N)$ range over fixed finite
sets independent of the stage of the construction. For a current network $N$
and an original vertex $x\in V(\mathcal H[\mathcal C])$, write
$m_N(x):=\bigl|\{\,v\in V(N)\mid \lambda(v)=x\,\}\bigr|$.

We claim that \(m_N(x)\le\nu(x)\)
for every original vertex $x$, where $\nu(x)$ is the number of directed Hasse
paths from the root $X$ to $x$. We prove this by induction over the Hasse
order. The root $X$ has exactly one current copy, so the claim holds for
$X$.

Let $x\neq X$, and assume the claim has been proved for all original parents
of $x$. For each current copy $v$ of $x$, choose one parent $q(v)$ in the
current network. By Proposition~\ref{prop:splitting-invariants}(c), the original
label $\lambda(q(v))$ is an original parent of $x$. The map $v\mapsto q(v)$ is
injective: if two distinct current copies of $x$ had the same chosen parent,
that parent would have two children with original label $x$, contradicting Claim 2. Therefore
\[
m_N(x)\le
\sum_{p\in\operatorname{Par}_{\mathcal H}(x)}m_N(p).
\]
By the induction hypothesis, this is at most
\(\sum_{p\in\operatorname{Par}_{\mathcal H}(x)}\nu(p)\).
Every directed Hasse path from $X$ to $x$ ends with a unique edge $p\to x$,
where $p\in\operatorname{Par}_{\mathcal H}(x)$, so the last sum is exactly
$\nu(x)$. This proves the claim.

Consequently, for every $i$ and every $\ell$,
\[
0\le a_{i,\ell}(N)
\le \sum_{\substack{x\in V(B_i)\\ d_{B_i}(x)=\ell}} m_N(x)
\le \sum_{\substack{x\in V(B_i)\\ d_{B_i}(x)=\ell}} \nu(x).
\]
Thus each coordinate of $\Theta(N)$ ranges over a fixed finite set. Since every
split strictly decreases $\Theta$ lexicographically, and since there are only
finitely many possible vectors with these coordinate bounds, the splitting
process must terminate. This proves (a), while Claim 1 proves (b).
\end{proof}

\section*{Statements and Declarations}

\noindent\textbf{Competing interests.} The authors declare that they have
no competing interests.\par

\noindent\textbf{Data availability.} No empirical data are analyzed in
this article; all examples are illustrative cluster systems defined in the
text.\par

\noindent\textbf{Declaration of generative AI and AI-assisted technologies in
the writing process.} During the preparation of this manuscript, the authors
used ChatGPT to assist with language polishing, readability improvements, and
organization of expository text on the basis of an existing draft. The
mathematical content, including the ideas, definitions, results, proofs,
constructions, and references, was provided and checked by the authors. After
using this tool, the authors reviewed and edited the manuscript as needed and
take full responsibility for its content.\par

\end{document}